\documentclass[journal,comsoc]{IEEEtran}

\IEEEoverridecommandlockouts
\usepackage[noadjust]{cite}
\usepackage{amsmath,amssymb,amsfonts}
\usepackage{graphicx}
\usepackage{textcomp}
\usepackage{xcolor}
\usepackage{amsmath}
\usepackage{amsthm}
\usepackage{subfig}
\usepackage{float}

\newtheorem{theorem}{\bf Theorem}

\newtheorem{lemma}{\bf Lemma}
\newtheorem{definition}{\bf Definition}
\newtheorem{remark}{\bf Remark}

\usepackage{algorithm}
\usepackage{algpseudocode}

\usepackage{bm}
\newcommand{\sK}{\mathsf{K}}
\newcommand{\sM}{\mathsf{M}}
\newcommand{\sP}{\mathsf{P}}
\newcommand{\sT}{1}
\newcommand{\ssT}{\mathsf{T}}
\newcommand{\sA}{\mathsf{A}}

\newcommand{\sTu}{\mathsf{T^{\text{u}}}}
\newcommand{\sTd}{\mathsf{T^{\text{d}}}}

\newcommand{\gmt}{\bg_{\mathcal{M}\backslash\{i\} } } 
\newcommand{\nmt}{\bm{n}_{\mathcal{M}\backslash\{i\}}}
\newcommand{\cM}{\mathcal{M}}
\newcommand{\cK}{\mathcal{K}}

\newcommand{\gD}{\bg_D}
\newcommand{\bg}{\bm{g}}
\newcommand{\bc}{\bm{c}}
\newcommand{\bn}{\bm{n}}
\newcommand{\bu}{\bm{u}}
\newcommand{\tbc}{\tilde{\bc}}
\newcommand{\bfY}{\mathbf{Y}}
\newcommand{\tsY}{\tilde{\mathbf{Y}}}
\newcommand{\bfX}{\mathbf{X}}
\newcommand{\bfH}{\mathbf{H}}
\newcommand{\bfPhi}{\mathbf{\Phi}}
\newcommand{\mbf}[1]{{\mathbf{#1}}}
\newcommand{\tn}[1]{{\textnormal{#1}}}
\newcommand{\tld}[1]{{\tilde{#1}}}

\begin{document}
	%
	\title{Asymptotically Optimal Secure Aggregation for Wireless Federated Learning with Multiple Servers}
	%
	%
	%
	
	\author{Zhenhao Huang, Kai Liang, Yuanming Shi, Songze Li, and Youlong Wu
			\thanks{This work was in part presented at the \textit{IEEE International Symposium on Information Theory (ISIT)}, 2023 \cite{huang2023sgawmsfl}. }
			\thanks{Z. Huang, K. Liang, Y. Shi and Y. Wu are with the School of Information Science and Technology,	ShanghaiTech University, Shanghai, China. Emails: \{Huangzhh,liangkai,shiym,wuyl1\}@shanghaitech.edu.cn.  }
			\thanks{S. Li is with the School of Cyber Science and Engineering, Southeast University, Nanjing, China. E-mail: songzeli@seu.edu.cn.}
		}
	
	%
	%
	\maketitle

	\begin{abstract}
		 In this paper, we investigate the transmission latency of the secure aggregation problem in a \emph{wireless} federated learning system with multiple curious servers.
		We propose a privacy-preserving coded aggregation scheme where the servers can not infer any information about the distributed users' local gradients, nor the aggregation value. 
		In our scheme, each user encodes its local gradient into $\sK$ confidential messages intended exclusively for different servers using a multi-secret sharing method, and each server forwards the summation of the received confidential messages, while the users sequentially employ artificial noise alignment techniques to facilitate secure transmission. 
		Through these summations, the user can recover the aggregation of all local gradients.  
		We prove the privacy guarantee in the information-theoretic sense and characterize the uplink and downlink communication latency measured by \emph{normalized delivery time} (NDT), both of which decrease monotonically with the number of servers $\sK$ while increasing over most of the range of the number of users $\sM$.  
		Finally, we establish a lower bound on the NDT of the considered system and theoretically prove that the scheme achieves the optimal uplink and downlink NDT under the conditions $\sK \gg \sM \gg 0$ and $\sK \gg \sM$, respectively.  For arbitrary $\sK$ and $\sM$, the proposed scheme achieves the optimal uplink NDT within a multiplicative gap of $4$.  
	\end{abstract}
	
	\begin{IEEEkeywords}
		Secure aggregation, coded computing, communication latency, interference alignment,  wireless federated learning.
	\end{IEEEkeywords}
	
	\IEEEpeerreviewmaketitle

	\section{Introduction}
	\label{sec:intro}
	Recently, various large-scale learning applications have shown their potential with the rapidly increasing volume of modern datasets.
	Due to the dispersion of datasets and the requirements of low-latency, traditional centralized computation frameworks are becoming incompetent \cite{Zhu2020Toward, shi2020communication, shi2023taskoriented, letaief2021edge}. Therefore, many computational tasks are preferably performed distributively and require collaboration among parties in terms of the computational power and the collection of datasets. 
	However, privacy concerns hold back the data sharing \cite{Ulukus2022private, ma2023trusted}. Aiming at addressing such challenge, a popular framework named Federated Learning (FL) has been introduced, where a learning model is trained collaboratively by multiple users while keeping their datasets local \cite{mcmahan2017communication,kairouz2019advances}.
	The typical FL system requires the users to cooperate with the server to train the desired global model. Specifically, the users share the model updates computed based on the local data, and the server, which is honest but curious, is responsible for aggregating the local updates of users, but is required to learn nothing about the users' local data. 
	Despite the training datasets being kept locally, it has been shown that the data can still be reconstructed from users' local updates \cite{Geiping2020Inverting,Wang2019beyond,zhu2019deep}. 
	
	To reduce the privacy leakage\footnote{The definition of security and privacy are ambiguous in many existing literature, for a comprehensive understanding of the security and privacy risks of distributed machine learning, we refer the reader to \cite{ma2023trusted}. In this paper, we aim to protect the gradient transmission, which is typically the focus of the secure aggregation problem.}, differential privacy (DP) via noise addition has been proven effective in protecting the local updates in FL \cite{ geyer2017differentially, sun2020ldp}. 
	Differential Privacy (DP) is a noisy release mechanism that protects individual data by injecting permanent noise into it before disclosing it to an untrusted data aggregator \cite{dwork2014algorithmic}.  
	In addition, for FL in wireless access networks, i.e., the so-called wireless FL, the DP mechanism can take advantage of the superposition property and the natural noise of the wireless channel \cite{seif2020wireless, liu2020privacy, elgabli2021harnessing, yang2022differentially}.
	The work in \cite{seif2020wireless} proposed an analog aggregation scheme for the FL over a flat-fading Gaussian multiple access channel under the local DP constraints, and show that the superposition property of the wireless channel brings benefits of efficiency and privacy. The work in \cite{liu2020privacy} studied the impact of the power allocation and channel noise in the gradient transmission on convergence and privacy. 
	Although DP provides a popular privacy analysis framework, the privacy guarantee provided by DP comes at the cost of degrading the learning utility since the aggregation values are not accurate.
	
	With better assurance in learning utility and even stronger privacy constraints, the concept of secure aggregation has gained attention in FL \cite{bonawitz2016practical,Bonawitz2017practical}. 
	The key idea of secure aggregation is that the local gradients are masked based on cryptography or coding theory before they are sent to the server, and these masks cancel out when a trusted server or user aggregates the masked messages, whereas an unauthenticated party learns little or nothing about the original gradients. 
	In a fascinating recent work \cite{Bonawitz2017practical}, the masks are generated based on a cryptographic protocol and stored by all users using the secret sharing method, so that the masks can be reconstructed even if some users drop out. 
	An additive homomorphic encryption based method is proposed to preserve the privacy of the gradients so that even the aggregated gradients are not leaked to the server \cite{aono2017privacy}.
	Secure aggregation protocols based on cryptography usually incur extensive communication and computation in the FL scenario, hence many works focus on improving the computation and communication efficiency of the secure aggregation protocol \cite{bell2020secure,kadhe2020fastsecagg,So2021turboaggregate,yang2021lightsecagg, Zhao2022information}. The work in \cite{kadhe2020fastsecagg} proposed a multi-secret sharing scheme based on the Fast Fourier Transform, which reduces the computational cost without increasing the communication cost. The work in \cite{So2021turboaggregate} designs a multi-group circular aggregation strategy that achieves a communication overhead of $O(N\log{N})$ and provides privacy against up to $N/2$ colluding users, where $N$ is the number of users.  
	Using tools from coding theory, \cite{Zhao2022information} designed an information-theoretic secure aggregation protocol that significantly reduces the aggregation complexity but relies on a trusted third party to prepare random masks for users.
	In spite of its benefits, the secure aggregation schemes mentioned above only focus on the wired system and typically do not consider the characteristics of wireless scenarios.
	
	{On the other hand, with the advancement of communication technology, such as 5G and 6G, the wireless FL, has received increasing attention due to its vast potential and broad range of applications, including the Vehicle Networking and the Internet of Things \cite{chen2021distributed}, \cite{amiri2020federated}.
	Note that the aforementioned works \cite{bonawitz2016practical,Bonawitz2017practical,bell2020secure,kadhe2020fastsecagg,So2021turboaggregate,yang2021lightsecagg, Zhao2022information,geyer2017differentially, sun2020ldp,seif2020wireless, liu2020privacy, elgabli2021harnessing, yang2022differentially}  would face an risk of information leakage due to the inherent broadcast nature of wireless communication, i.e., wireless broadcast signals are received indiscriminately by all nodes, unlike wired communication where signals are sent to a specific target node.  For instance, exchanging gradients in wireless FL is vulnerable to an eavesdropping attack, in which a malicious actor may collect radio signals to infer private information for illegal purposes \cite{zhang2022wireless,ma2020safeguarding,wei2020federated}.
	Moreover, in the existing works  \cite{bonawitz2016practical,Bonawitz2017practical,bell2020secure,kadhe2020fastsecagg,So2021turboaggregate,yang2021lightsecagg, Zhao2022information,geyer2017differentially, sun2020ldp,seif2020wireless, liu2020privacy, elgabli2021harnessing, yang2022differentially}, the server is allowed to learn the aggregated global gradients, while the aggregation value may still be sensitive and leak individual information about users \cite{so2023securing} \cite{elkordy2023much}. 
	
	To address the issues mentioned above, this paper investigates the transmission latency of the secure aggregation problem over a \emph{wireless} network. We consider the secure aggregation in a wireless FL system with multiple curious servers, where the use of multiple servers has been shown helpful in protecting the privacy of uplink and downlink in the federated learning \cite{JiaXTfsl,liang2024privacy}.
	The main contributions are summarized as follows. 
	\begin{itemize}
		\item We propose a privacy-preserving scheme where each server cannot infer the local gradients, including the aggregation value. 
		In our scheme, each server aggregates the secret shares generated based on different local gradients so that they are oblivious to the original gradients, while the user can recover the aggregation of the gradients from the aggregated shares of different servers due to the linearity of the adopted secret sharing method. To reduce the information obtained by the servers through overhearing on the wireless channel, users will transmit the artificial noise in turns. These methods not only preserve the privacy but also reduce the communication latency.
		\item We characterize the uplink and downlink communication latency, measured by the \emph{normalized delivery time} (NDT), and prove the privacy guarantee in the information-theoretic sense. Both the uplink and downlink NDTs decrease monotonically with the number of servers $\sK$. The uplink NDT increases with the number of users $\sM$ when $\sM \geq \sqrt{\sK} + 1$, while the downlink NDT always increases with the number of users $\sM$. Compared to the single-server system, increasing the number of servers in our scheme can significantly reduce the uplink NDT with a relatively small growth of the downlink NDT.
		\item We also establish an information-theoretic lower bound on both the uplink and downlink NDTs for the considered system, and theoretically prove that the obtained uplink NDT is within a multiplicative gap of $4$ from the optimum. Furthermore, the uplink NDT is asymptotically optimal when $\sK \gg \sM \gg 0$, and the downlink NDT is asymptotically optimal when $\sK \gg \sM$.
	\end{itemize}

{This work is in part based on \cite{huang2023sgawmsfl}, where the lack of clarity in the proof of decodability led to misleading conclusions. In this version, we design a new transmission for the downlink\footnote{The conference version investigated the problem under the assumption of colluding servers, while also lacking an analysis of optimality.}. Additionally, we derive an information-theoretic lower bound on the uplink and downlink latency of the system and prove that the achievable latencies are asymptotically optimal.}
	
	\emph{Notations}: We use the sans serif font for constants, bold for vectors  and matrices, and calligraphic font for most sets. Let $\mathbb{N}^{+}$ denote the set of positive integers. The sets of complex numbers are denoted by $\mathbb{C}$. Let $\mathcal{CN}(m,\delta^2)$ denote the complex Gaussian distribution with a mean of $m$ and variance of $\delta^2$. The operator $|\cdot|$ is the cardinality of a set or the absolute value of a scalar number. For any $k\in\mathbb{N}^{+}$, define $[k]=\{1,2,...,k\}$. For sets $\mathcal{S}$ and $\mathcal{Q}$, we let $\mathcal{S}\backslash\mathcal{Q} = \{i: i\in\mathcal{S}, i\notin\mathcal{Q}\}$.
	The operation $(\cdot)^{T}$ denotes transposition.
	
	\section{System Model and Problem Formulation}
	\label{sec:formulation}
	\begin{figure} [htbp]
		\centering
		\subfloat[uplink]{
			\begin{minipage}[t]{\linewidth}
				\centering
				\includegraphics[width=0.95\textwidth]{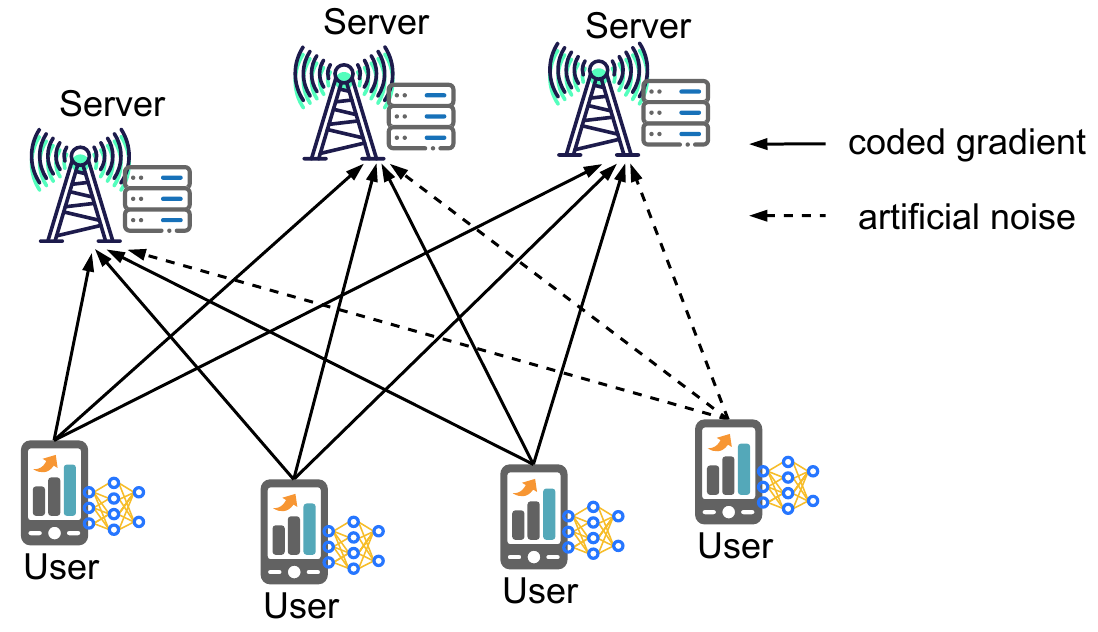}
			\end{minipage}
		}
		
		\subfloat[downlink]{
			\begin{minipage}[t]{\linewidth}
				\includegraphics[width=0.95\textwidth]{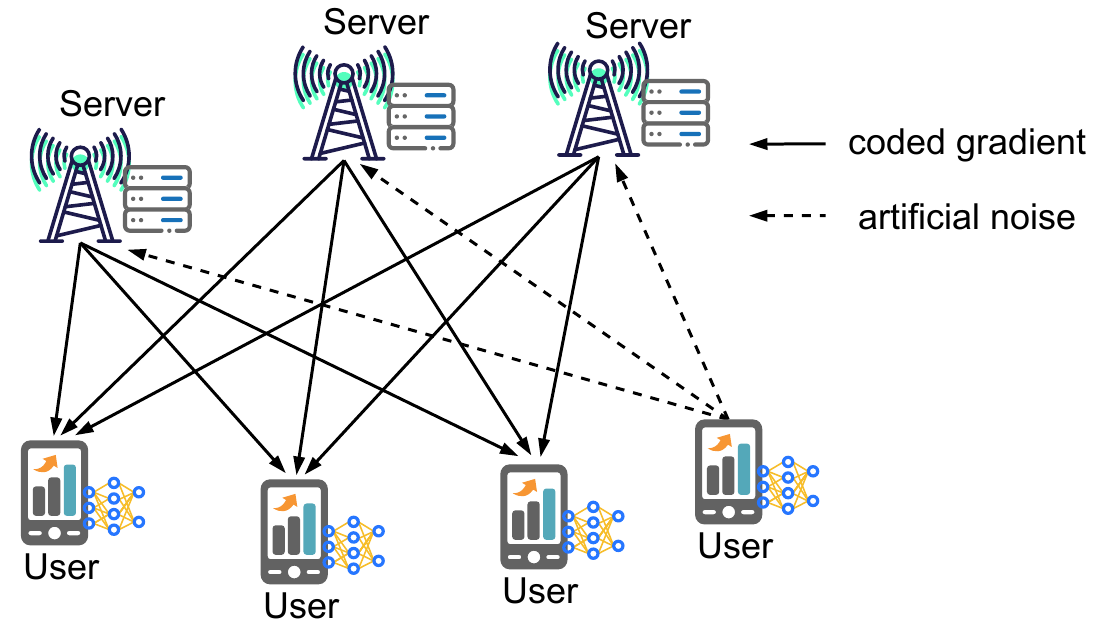}
			\end{minipage}
		}
		\caption{\label{sysmodel} Aggregation over wireless network with multi-server.}
	\end{figure}
	Consider a multi-server FL system as shown in Fig. \ref{sysmodel}. There exist a group of users, denoted by $\cM$ and indexed by $\{1,2,...,\sM\}$,  and a group of servers  denoted by $\cK$ and indexed by $\{1,2,...,\sK\}$, with $\sM=|\cM|$ and $\sK=|\cK|$. 
 Each user has a local private gradient $\bg_i\in\mathbb{F}^p$, where $\mathbb{F}^p$ is a vector space of dimension $p$ over a finite field $\mathbb{F}$. Each server is assumed to be semi-honest, following the protocol truthfully, but trying to gain additional information about the local $\bg_i$ and the aggregation $\gD = \sum_{i=1}^{\sM}\bg_i$ by using the information sent by the users.
	The users want to obtain the aggregation $\gD$ with the help of the servers, while the servers are unaware of the local gradient $\bg_i$ and the aggregation $\gD$ (see formal definition in Definition \ref{DefPrivacy}).  Assume that the channel state information (CSI) is perfectly known to all users and servers.
	
	For the uplink transmission, each user $i\in\cM$ first encodes the local gradient $\bg_i$ and some local randomness $\bn_i$ into a message $\bc_{j,i}\in \mathcal{W}_{j,i}$ destined for server $j$, i.e.,
	\begin{IEEEeqnarray}{rCl}
		(\bc_{1,i},...,\bc_{\sK,i}) = \phi_i(\bg_i,\bn_i), \nonumber
	\end{IEEEeqnarray}
	where $ \mathcal{W}_{j,i}$ is the corresponding message set with size of $\log | \mathcal{W}_{j,i}|$ bits representing and $\phi_i$ is an appropriate encoding function. 
	Let the uplink signal transmitted by user $i$ at time slot $t\in[\sTu]$ be ${X}_{i}(t)$, where $\sTu$ is the total channel uses of the uplink transmission. Based on these messages, the user $i\in\cM$ generates the complex channel inputs
	\begin{IEEEeqnarray}{rCl}
		\bfX_i \triangleq ( X_i(1),...,X_i(\sTu) ) = f_i^{(\sTu)}(\{\bc_{j,i}\}_{j\in\cK}),\nonumber
	\end{IEEEeqnarray}
	by an encoding function $f_i^{(\sTu)}$ on appropriate domains and the input signals satisfy an average power constraint $\sP$, i.e., $\frac{1}{\sTu} \sum_{t=1}^{\sTu} E[ | {X}_i(t)|^2 ] \leq \sP$. 
	At time slot $t$,  each server $j\in\cK$ observes a linear combination of the signals sent by all users, i.e.,
	\begin{IEEEeqnarray}{rCl}
		{Y}_j(t) = \sum_{i=1}^{\sM}  {H}_{j,i}(t)  {X}_{i}(t) + {Z}_j(t)\nonumber, 
	\end{IEEEeqnarray}
	where  $ {Z}_j(t)\sim \mathcal{CN}(0,1)$ denotes the additive white Gaussian noise (AWGN) at server $j$ and $ {H}_{j,i}(t)\in \mathbb{C}$ denotes the channel coefficient from user $i$ to the server $j$. All the channel coefficients and noise are assumed to be identically and independently (i.i.d.) distributed across the time and users.
	
	For the downlink transmission, based on the received signals   
	$\bfY_j  \triangleq (  {Y}_{j}(1), ... ,  {Y}_{j}( \sTu) )$, each server $j\in \cK$  decodes desired messages and generates  downlink messages  $\tbc_{i,j}\in \tilde{\mathcal{W}}_{i,j}$ intended for user $i$, where $ \tilde{\mathcal{W}}_{i,j}$ is the corresponding message set with size of $\log | \tilde{\mathcal{W}}_{i,j}|$ bits representing.  
	Let the downlink signal sent by server $j$ at time slot $t\in[\sTd] $ be $\tilde{ {X}}_{j}(t)$, where $\sTd$ is the total channel uses of the downlink transmission. 
	Based on these messages, server $j\in\cK$ produces the complex channel inputs
	\begin{IEEEeqnarray}{rCl}
		\tilde{\bfX}_j \triangleq ( \tilde{X}_j(1),...,\tilde{X}_j(\sTd) ) = \tilde{f}_j^{(\sTd)}(\{\tbc_{i,j}\}_{i\in\cM}),\nonumber
	\end{IEEEeqnarray}
	by an encoding function $\tilde{f}_j^{(\sTd)}$ on appropriate domains and the input signals satisfy an average power constraint $\sP$, i.e., $\frac{1}{\sTd} \sum_{t=1}^{\sTd} E[ |\tilde{ {X}}_j (t)|^2 ] \leq \sP$.
	The downlink signal at user $i\in\cM$ at time slot $t$, denoted by $\tilde{ {Y}}_i(t) \in \mathbb{C}$, is modeled by
	\begin{align*}
		\tilde{ {Y}}_i (t)\!=\! \sum_{j=1}^{\sK} \tilde{ {H}}_{i,j}(t) \tilde{ {X}}_j (t) \!+\! \tilde{ {Z}}_i (t),
	\end{align*}
	where $\tilde{ {Z}}_i (t)\sim \mathcal{CN}(0,1)$ denotes the noise at user $i$ and $\tilde{ {H}}_{i,j}(t)\in \mathbb{C}$ is the downlink channel coefficient from server $j$ to user $i$.
	
	
	After receiving $\tsY_i \triangleq ( \tilde{ {Y}}_i(1), ... , \tilde{ {Y}}_i(\sTd) )$, user $i$ decodes the messages $\{\tbc_{i,j}$, $j\in\cK\}$ intended for it, i.e.,
	\begin{IEEEeqnarray}{rCl}
		(\tbc_{i,1},...,\tbc_{i,\sK}) = \psi_i(\tsY_i)\nonumber.
	\end{IEEEeqnarray}
	where $\psi_i$ is an appropriate decoding function.
	Finally, the user $i$  produces an estimation  $\hat{\bg}_D$ based on $(\tbc_{i,1},...,\tbc_{i,\sK})$ and the local gradient $\bg_i$, where $\hat{\bg}_D$ is the estimation of ${\bg}_D$ satisfying  $\lim_{\sTu,\sTd\to \infty} \text{Pr}(\hat{\bg}_D\neq {\bg}_D) = 0$. 
	This constraint can also be written as 
	\begin{IEEEeqnarray}{rCl}\label{cons:correct}
		\lim\limits_{\sTu,\sTd\to \infty}H(\gD|(\tbc_{i,1},...,\tbc_{i,\sK}), \bg_i) = 0.
	\end{IEEEeqnarray}    
	
	We introduce the following definitions to formulate the communication latency and privacy-preserving constraints. Before that, we first define the following notations: 
	$\bc_{\cK,\cM}$ and $\tbc_{\cM,\cK}$ denote the uplink and downlink message sets, respectively, i.e.,  $\bc_{\cK,\cM}=(\bc_{j,i}:j\in\cK,i\in\cM) $ and $\tbc_{\cM,\cK}=(\tbc_{i,j}:i\in\cM,j\in\cK) $.  
	\begin{definition}\emph{Degree of freedom (DoF)}.\label{DefDof}
		Define the degree of freedom (DoF) of the uplink and downlink transmissions as   $d^{\text{up}}_{\text{sum}} = \lim_{\sP\to\infty} \frac{ H( \bc_{\cK,\cM} )} {\sTu \log (\sP)}$ and  $d^{\text{down}}_{\text{sum}} = \lim_{\sP\to\infty} \frac{H( \tbc_{\cM,\cK})  } { \sTd\log (\sP)}$, respectively.
	\end{definition} 
	\begin{definition}\emph{Normalized delivery time (NDT)\cite{sengupta2016cacheaided}}. \label{DefNDT}
		Let $\sA\triangleq  \log(|\mathbb{F}^p|)$ denote the bit length of each local gradient. The normalized uplink delivery time is defined as 
			$\bm{\Delta}^{\text{up}} = \lim\limits_{\sP,\sA\to\infty} \frac{\sTu}{\sA/\log{(\sP)}}$.
		The normalized downlink delivery time is defined as 
			$\bm{\Delta}^{\text{down}} = \lim\limits_{\sP,\sA\to\infty} \frac{\sTd}{\sA/\log{(\sP)}}$.
	\end{definition} 
	
	\begin{definition}\label{DefPrivacy}
		\emph{Privacy}.
		The equivocation for local gradients $\bg_\cM$ at a server $j\in\cK$ is defined as
		\begin{align}\label{def:equ_j}
			\Delta^{j}_{\bg_{\cM}} \triangleq \frac{ H(\bg_{\cM} | \bfY_{j}, \tilde{\bfY}_{j}) }{ H(\bg_{\cM}) }, 
		\end{align}
		where $\bg_{\cM} = \{ \bg_1, \bg_2, ... , \bg_\sM \}$ and $\bfY_{j}, \tilde{\bfY}_{j}$ are the received signals of the server $j$ from uplink and downlink transmission. 
		We say that the scheme is \emph{private} if
		\begin{align}
			\lim_{\sA,\sP\to\infty} \Delta^{j}_{\bg_{\cM}} &= 1-\epsilon, \forall j\in\cK, 
		\end{align}
		for arbitrary small positive constant $\epsilon$.
		
	\end{definition}
	
	The equation \eqref{def:equ_j} indicates the level at which the curious server is confused. To prove the \emph{privacy} guarantee, we assume that all $\bg_i$ are independent, and each follows a uniform distribution over the field. 
	The uniformity and independence of the gradients are required for the converse proof, but are not necessary for the achievability scheme \cite{Zhao2022information}. 
	A lower bound on the decoding error probability for $\bg_{\cM}$ at the server can be derived from the equivocation as follows \cite{he2016secrecy}. From Fano's inequality, we know \begin{align}
			P_e\geq \frac{H(\bg_{\cM}|\bfY, \tld{\bfY}) - 1}{H(\bg_{\cM})} = \Delta_{\bg_{\cM}} - \frac{1}{H(\bg_{\cM})},
		\end{align}
	Furthermore, when the entropy of the gradients is very large, it has $\lim\limits_{\sA\to\infty}P_e\geq  \Delta_{\bg_{\cM}} - \lim\limits_{\sA\to\infty}\frac{1}{H(\bg_{\cM})} = \Delta_{\bg_{\cM}}$. Thus, $P_e$ is asymptotically lower bounded by $\Delta_{\bg_{\cM}}$.
	Definition \ref{DefPrivacy} indicates that any server $j$ cannot infer the local gradients, including the aggregation value. 
	
	Our goal is to design the aggregation and communication schemes with as little NDT as possible, so that all users can recover the aggregation value $\bg_D$ with $\lim\limits_{\sTu,\sTd\to \infty} \text{Pr}(\hat{\bg}_D\neq {\bg}_D) = 0$ and satisfying the privacy guarantee.

	\section{Wireless Coded Secure Gradient Aggregation}
	\label{sec:wlcsga}
	In this section, we introduce a wireless secure gradient aggregation scheme for the system described in Section \ref{sec:formulation}. First, a simple example is given to illustrate the operation of the aggregation scheme, then we give the general description and the achievable NDT of the proposed scheme. 
	
	\subsection{An Illustrative Example}\label{illutrating example}
	Consider a distributed FL system consisting of $\sM=5$ users and $\sK=4$ servers. Each user wants the sum of the local gradients of all users, i.e., $\bg_D = \sum_{i=1}^{5} \bg_i$ with $\bg_i\in \mathbb{F}^{ {p}}$ calculated by user $i$ from its local data. 
	
	Each user $i\in [5] $ first splits $\bg_i$ into $r = 3$ parts, i.e., 
	$\bg_i = ( \bg_{i, 1}, \bg_{i, 2}, \bg_{i, 3} )$, and encodes $\bg_i$ using a Lagrange interpolation polynomial:
	\begin{align}\label{ex_encoding_gi}
		\mathcal{G}_i(x) &=  \bg_{i,1}\! \frac{(x\!\!-\!\!2)(x\!\!-\!\!3)(x\!\!-\!\!4)}{(1\!\!-\!\!2)(1\!\!-\!\!3)(1\!\!-\!\!4)} \!+\! \bg_{i,2}\! \frac{(x\!\!-\!\!1)(x\!\!-\!\!3)(x\!\!-\!\!4)}{(2\!\!-\!\!1)(2\!\!-\!\!3)(2\!\!-\!\!4)} \nonumber\\
		&+ \bg_{i,3}\! \frac{(x\!\!-\!\!1)(x\!\!-\!\!2)(x\!\!-\!\!4)}{(3\!\!-\!\!1)(3\!\!-\!\!2)(3\!\!-\!\!4)} +\!\bn_{i} \frac{(x\!\!-\!\!1)(x\!\!-\!\!2)(x\!\!-\!\!3)}{(4\!\!-\!\!1)(4\!\!-\!\!2)(4\!\!-\!\!3)}, 
	\end{align}
	where  $\bn_{i}$ is chosen uniformly at random from $\mathbb{F}^{ \frac{p}{2} }$. Let $\mathcal{F}(x)\triangleq \sum_{i=1}^5 \mathcal{G}_i(x)$, which is another polynomial function of degree $2$.  Note that  $\mathcal{F}(x)$ has a nice property: $\mathcal{F}(1)=\sum_{i=1}^5 \bg_{i,1}$, $\mathcal{F}(2)=\sum_{i=1}^5 \bg_{i,2}$ and $\mathcal{F}(3)=\sum_{i=1}^5 \bg_{i,3}$, which are exactly the sums of local gradient parts. 
	
	We then select $4$ distinct elements $\{\alpha_j\}_{j=1}^4$ in $\mathbb{F}$ such that $\{\alpha_j\}_{j=1}^4\cap [4]=\emptyset$. User $i$ generates the $\sK=4$ \emph{confidential} message $\bc_{j,i}\in \mathbb{F}^{ \frac{p}{2} }$ exclusively intended for server $j$ according to  $\bc_{j,i} = \mathcal{G}_i (\alpha_j)$, for $j\in[4]$.  
		
	  To send all confidential messages each user $i$ further split $\bc_{j,i}$ into $4$ segments of equal size, i.e., $\bc_{j,i}=\{c_{j,i}^{1}, c_{j,i}^{2}, c_{j,i}^{3}, c_{j,i}^{4}\}$. 
	All segments will be sent through a total of $5$ rounds of transmission each containing $\mathsf{T}^{\prime}=4n^{\Gamma}+4(n+1)^{\Gamma}$ channel uses (i.e., $\sTu = 5\mathsf{T}^{\prime}$), where $\Gamma = 12$ (determined by the alignment constraints in \eqref{eqAlgnCon}). In round $i\in[5]$,  user $i$   sends an artificial noise while the other users deliver their segments to desired servers. 
	Taking   round  5 as an example,
	the user $i\in[4]$ produces a complex channel input  of the form
	\begin{align*}
		\bfX_i = \sum_{j=1}^{4} \bfPhi^{[j,i]} \bu_{j,i}^{4}, 
	\end{align*}
	where  $\bu_{j,i}^{4}$ is an $n^{\Gamma}\times 1$ symbol vector encoded from $\bc_{j,i}^{4} $, and $\bfPhi^{[ji]}$ is a $\sTu\times n^{\Gamma}$ beamforming matrix. 
	Meanwhile, the user $5$ transmits the signal  
	\begin{align*}
		\bfX_5 = \sum_{j=1}^{4} \bfPhi^{[j,5]}\bm{v}_{j,5},
	\end{align*}
	where $\bm{v}_{j,5}$ is the $(n+1)^{\Gamma}\!\!\times\!1$ artificial noise symbol vector chosen from the Gaussian distribution $\mathcal{CN}(0,\frac{P}{(n+1)^{\Gamma}}\mathbf{I}_{(n+1)^{\Gamma}} )$, and $\bfPhi^{[j,5]}$ is the $\sTu\times (n+1)^{\Gamma}$ beamforming matrix. 
	For server $k$, the received signal is 
	\begin{align*}
		\bfY_k = \sum_{i=1}^{4} \bfH_{k,i} ( \sum_{j=1}^{4} \bfPhi^{[j,i]} \bu_{j,i} ) + \bfH_{k,5}\sum_{j=1}^{4}\bfPhi^{[j,5]}\bm{v}_{j,5} + \bm{Z}_k,
	\end{align*}
	where $\bfH_{k,i}$ is a $\sTu\times\sTu$ diagonal matrix with the diagonal elements being the channel coefficient. 
	
	\begin{figure*}[htpb]
		\centering
		\includegraphics[width=\linewidth]{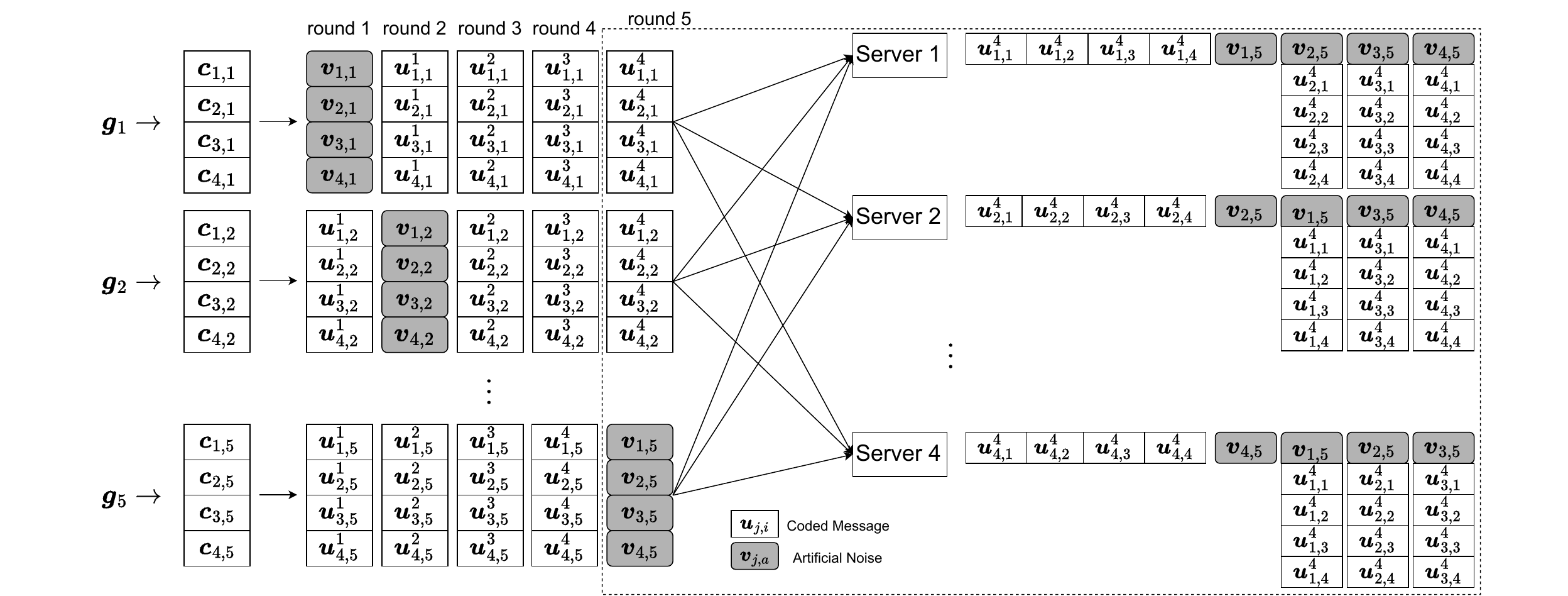}
		\caption{The uplink transmission over the wireless network. For server $1$, the confidential messages $\{\boldsymbol{u}^{4}_{1,1},...,\boldsymbol{u}^{4}_{1,4}\}$ and an artificial noise $\boldsymbol{v}_{1,5}$ occupy independent dimensions in the received signal space, while the messages intended for servers $2$, $3$ and $4$ are aligned to the subspace spanned by the artificial noise $\boldsymbol{v}_{2}$, $\boldsymbol{v}_{3}$ and $\boldsymbol{v}_{4}$, respectively. The alignment conditions are similar for the other servers.}
		\label{example_5x3}
	\end{figure*}
	
	According to \cite{wang2015secure}, there exist beamforming matrices $\bfPhi^{[j,i]}$ that satisfy the following alignment conditions such that the messages destined for server $j\in[4]$ are aligned to the subspace spanned by the artificial noise $\bm{v}_{j,5}$ at the server $k\in[4]\backslash\{j\}$, i.e., 
	\begin{align}\label{eqAlgnCon}
		\begin{cases}
			\bfH_{k,1}\bfPhi^{[j,1]} \prec  \bfH_{k,5}\bfPhi^{[j,5]} \\
			\bfH_{k,2}\bfPhi^{[j,2]} \prec  \bfH_{k,5}\bfPhi^{[j,5]} \\
			\bfH_{k,3}\bfPhi^{[j,3]} \prec  \bfH_{k,5}\bfPhi^{[j,5]} \\
			\bfH_{k,4}\bfPhi^{[j,4]} \prec  \bfH_{k,5}\bfPhi^{[j,5]} 
		\end{cases} \qquad \forall k\in[4]\backslash\{j\},
	\end{align} 
	which contains $\Gamma = 12$ relations. 
	At the same time, the desired messages for server $k$ occupy the signal space $\bfH_{k,1}\bfPhi^{[k,1]},...,\bfH_{k,4}\bfPhi^{[k,4]}$ which are mutually independent and independent of the artificial noise.
	The transmission is illustrated in Fig.~\ref{example_5x3},  and the transmission in other rounds follows the same way.

	 At the transmission of round 5, there are $4$ confidential messages from $\{\bc_{k,1},\bc_{k,2},\bc_{k,3},\bc_{k,4}\}$ for server $k$, which occupies a total of $4n^{\Gamma}$ independent dimensions. Pulsing $4(n+1)^{\Gamma}$ dimensions spanned by $\bfH_{k,5}\bfPhi^{[1,5]}\bm{v}_{1,5},...,\bfH_{k,5}\bfPhi^{[4,5]}\bm{v}_{4,5}$, the overall $\sTu = 4n^{\Gamma} + 4(n+1)^{\Gamma}$ dimensions are needed to successfully deliver the confidential messages. Thus, the sum-DoF of this channel is 
		$d^{\text{up}}_{\text{sum}} = \frac{ 5\times4\times 4n^{\Gamma} }{ 5(4n^{\Gamma} + 4(n+1)^{\Gamma}) } \approx 2$.
	On the other hand, according to the Definition \ref{DefDof}, we also obtain  $d^{\text{up}}_{\text{sum}}=\lim_{P\to\infty}\frac{\sA/3 \cdot 20 }{ \sTu \log (P)}$. Therefore, we can get the NDT of the uplink as follows:
		$\Delta^{\text{up}} = \lim_{\sA,\sP\to\infty} \frac{\sTu}{\sA/\log{\sP}} = \frac{20/3}{ d^{\text{up}}_{\text{sum}}} = \frac{10}{3}$.
	
	Next, each server $j\in[4]$ first recovers confidential messages $\{\bc_{j,1}, ..., \bc_{j,4}\}$ based on received signals, and then computes the sum of these confidential messages
	\begin{align*}
		\sum_{i=1}^{5} \bc_{j,i} = \mathcal{F}(\alpha_j)= \mathcal{G}_1(\alpha_j) + \mathcal{G}_2(\alpha_j) +\mathcal{G}_3(\alpha_j)  +\mathcal{G}_4(\alpha_j).
	\end{align*} 
	Therefore, each server $j$ owns one evaluation of the polynomial function $\mathcal{F}(x)$. The server $j$ generates the confidential messages $\tbc_{1,j}=\tbc_{2,j}=\tbc_{3,j}=\tbc_{4,j} = \mathcal{F}(\alpha_j)$. 
	For the downlink transmission, the same method as for the uplink is used to successfully send the confidential messages. 
	Note that the user node continues to sequentially transmit the artificial noise. Similarly, the sum-DoF of the downlink channel is $d^{\text{down}}_{\text{sum}} = \frac{ 5\times4n^{\Gamma} }{ 5(4n^{\Gamma} + 4(n+1)^{\Gamma}) } \approx 1/2$. Also, from the Definition \ref{DefDof}, we have $d^{\text{down}}_{\text{sum}}=\lim_{P\to\infty}\frac{\sA/3 \cdot 4 }{ \sTu \log (P)}$ and the NDT of the downlink is 
		$\Delta^{\text{down}} = \lim_{A,P\to\infty} \frac{T^{\text{down}}}{A/\log{P}} = \frac{4/3}{ d^{\text{down}}_{\text{sum}}} = \frac{8}{3}$.
	
	After the downlink transmission, each user attains $4$ distinct points of the polynomial $\mathcal{F}(x)$ of degree $3$, i.e., $(\alpha_1, \mathcal{F}(\alpha_1)), (\alpha_2, \mathcal{F}(\alpha_2)), (\alpha_3, \mathcal{F}(\alpha_3), (\alpha_4, \mathcal{F}(\alpha_4))$, and thus can    construct the function $\mathcal{F}(x)$ by interpolation.  Finally, every user  recovers the global gradient $\bg_D$ from the evaluations $\mathcal{F}(1)=\sum_{i=1}^5 \bg_{i,1}$, $\mathcal{F}(2)=\sum_{i=1}^5 \bg_{i,2}$ and $\mathcal{F}(3)=\sum_{i=1}^5 \bg_{i,3}$.
	The privacy guarantee is achieved by using artificial noise to confuse the decoding of confidential messages at the servers, and by masking the gradients with local randomness, rendering the original gradients undecipherable to the curious servers. See the Section \ref{subsec:priv} for more details.
	
	\subsection{General Description}
	\label{general description}
	The main idea is as follows. Each user $i\in[\sM]$ first splits each local gradient $\bm{g}_i$ into $r$ segments $\{\bm{g}_{i,k}\}_{k\in[r]}$ and encodes them into $\sK$ confidential messages based on a variate of multi-secret sharing method \cite{blakley1985security,franklin1992communication}, where each confidential message is exclusively intended for one server. In the downlink, each server sends a sum of its own confidential messages, which must be kept secret from others except the users. To ensure that the confidential messages are known only to the assigned one over the wireless channel, a user will send artificial noise during the message transmission based on an artificial noise alignment approach.  After decoding all the sums, each user uses Lagrange interpolation to recover $\sum_{i=1}^\sM\bm{g}_{i,1},\ldots,\sum_{i=1}^\sM\bm{g}_{i,r}$, respectively, to obtain the sum $\sum_{i=1}^\sM\bm{g}_i$.
	\begin{remark}
	The related work like artificial noise alignment \cite{wang2015secure} focuses on pure wireless communication problem that aims to directly deliver source messages instead of recovering a linear aggregation, and protecting the original source messages. In our setting,  it involves both uplink and downlink transmission, aims at aggregating information from distributed users, and protects both source messages and the aggregated values. As a result, new designs are required to incorporate the above techniques in a manner suited to the problem at hand. 
	\end{remark}
		
	Similar to the example in the section \ref{illutrating example}, each user $i$ first splits $\bg_i$ into $r$ segments, i.e., $\bg_{i} = ( \bg_{i,1},...,\bg_{i,r} )$, and encodes the $\{\bg_{i,k}\}_{k=1}^r$ through a multi-secret sharing method, where the gradient segments are referred to as the secrets.  
	Specifically, we used the following Lagrange interpolation polynomial \eqref{gen_encode_Gi} to encode the gradients, which is more natural and recently known as Lagrange coding \cite{yu2019lagrange}:
	\begin{align}\label{gen_encode_Gi}
		\mathcal{G}_i (x) = \sum_{k=1}^{r} \bg_{i,k} \cdot\!\!\!\!\!\! \prod_{l=1}^{[r+1]\backslash\{k\}}\!\!\!\! \frac{x\!\!-\!\!\beta_l}{\beta_k \!\!-\!\! \beta_l} + \bn_{i}  \cdot\!\!\!\!\!\! \prod_{l=1}^{[r+1]\backslash\{k\}}\!\!\!\! \frac{x\!\!-\!\!\beta_l}{\beta_{r+1}\!\!-\!\!\beta_l},
	\end{align}
	where $r$ is a designable parameters satisfying the relations $r+1\leq \sK$. $\beta_1, \ldots,\beta_{r+1}$ are distinct elements of $\mathbb{F}$ and $\bn_{i}$ is chosen uniformly at random from $\mathbb{F}^{ \frac{p}{r} }$. 
	Define another polynomial $\mathcal{F} (x)=\sum_{i\in\cM}\mathcal{G}_i (x)$, which is of degree $r$ and satisfies the property
	\begin{IEEEeqnarray}{rCl}\label{gen_encode_Fi}
		\mathcal{F} (\beta_k) = \sum_{i\in\cM}\mathcal{G}_i (\beta_k)=\sum_{i\in\cM} \bg_{i,k},\ \forall k\in[r].
	\end{IEEEeqnarray}
	We then select $\sK$ distinct elements $\alpha_1, ... , \alpha_\sK$ from $\mathbb{F}$ such that $\{ \beta_{l} \}_{l\in[r+1]}  \cap \{\alpha_j\}_{j\in\cK} = \emptyset$. User $i$ generates the encoded \emph{confidential} message exclusively intended for server $j$, i.e., $\bc_{j,i} = \mathcal{G}_i (\alpha_j)$. 
	The coding process can be represented by the following matrix formulation,
	\begin{align}\label{largrang_matrix}
		\begin{bmatrix}
			\bc_{1,1} \!\!&\!\! \!\!\cdots\!\! \!\!&\!\! \bc_{1,\sM} \\
			\vdots \!\!&\!\! \!\!\ddots\!\! \!\!&\!\!\vdots \\
			\bc_{\sK,1} \!\!&\!\! \!\!\cdots\!\! \!\!&\!\! \bc_{\sK,\sM}
		\end{bmatrix} 
		\!&=\!\underbrace{\begin{bmatrix}
				U_{1,1}  \!\!&\!\! \!\cdots\! \!\!&\!\! U_{1,r+\sT} \\
				\vdots \!\!&\!\! \!\ddots\! \!\!&\!\!\vdots \\
				U_{\sK,1} \!\!&\!\! \!\cdots\! \!\!&\!\! U_{\sK,r+\sT}
		\end{bmatrix}}_{U} \underbrace{\begin{bmatrix}
				\bg_{1,1} \!\!&\!\! \cdots \!\!&\!\! \bg_{\sM,1} \\
				\vdots \!\!&\!\! \ddots \!\!&\!\!\vdots \\
				\bg_{1,r} \!\!&\!\! \cdots \!\!&\!\! \bg_{\sM,r} \\
				\bn_{1,\sT} \!\!&\!\! \cdots \!\!&\!\! \bn_{\sM,\sT}
		\end{bmatrix}}_{G} \nonumber\\
		&=\begin{bmatrix}
			U^{L} \  U^{R}
		\end{bmatrix}\begin{bmatrix}
			\tilde{\bg} \\
			\tilde{\bn}
		\end{bmatrix},
	\end{align}
	where $U_{j,k} = \prod_{l\in[r+\sT]\backslash\{k\}}\frac{\alpha_j - \beta_{l}}{\beta_k - \beta_{l}}$, $U^{L}$ denotes the first $r$ columns of the $U$, $U^{R}$ denotes the last columns of the $U$, $\tilde{\bg}$ denotes the first $r$ rows of the $G$, and $\tilde{\bn}$ denotes the last rows of the $G$.
	
	Based on these messages, each user will generate the channel inputs and send them to different servers through a wireless interference network. The messages have assigned receivers and would like to be kept private from others. 
	In other words, due to the encoding process of the gradients, the considered network is converted as a $\sM\times \sK$ X-network with confidential messages \cite{wang2015secure}\cite{xie2014secure}. 
	The DoF of the $\sM\times\sK$ X-network with confidential messages is characterized as the Lemma \ref{lemma_dof_1} below. 
	\begin{lemma}\label{lemma_dof_1}
		(Theorem 2 \cite{wang2015secure}) An achievable sum secure degree of freedom (DoF) of the $\sM\times \sK$ X-network with confidential messages with time/frequency-varying channels is 
		\begin{IEEEeqnarray}{rCl}
			d(\sM,\sK) = \begin{cases}\label{dof_eh}
				\frac{\sK(\sM-1)}{\sK+\sM-2} & \sK=2, \\
				\frac{\sK(\sM-1)}{\sK+\sM-1} & \sK\geq 3.
			\end{cases}
		\end{IEEEeqnarray}
	\end{lemma}
	Next, we introduce how to apply the artificial noise alignment to support the wireless secure gradient aggregation. We explain it in the case of $\sK\geq 3$, and that of $\sK=2$ is analogous and simpler with minor modifications.
	A total of $\sM\cdot \sK$ confidential messages will be sent through $\sM$ rounds transmissions. Every user will take turns sending the artificial noise instead of useful messages in one round of transmissions.
	For the uplink transmission, each $\bm{c}_{j,i}$ is splited into $\sM-1$ disjoint segments each of equal size, i.e., $\bm{c}_{j,i} = \{ \bm{c}^{1}_{j,i}, \cdots \bm{c}^{(\sM-1)}_{j,i} \}$. 
	The users will send these segments through $\sM-1$ rounds out of $\sM$ rounds sequentially.
	Each coded segment $\bm{u}_{j,i}^{\tau}$ is sent over a block of $\ssT^{\prime} = \sK(n+1)^{\Gamma}+(\sM-1)n^{\Gamma}$ channel uses, where $\Gamma = (\sM-1) (\sK-1)$.

	In round $a$, user $a$ is chosen to send the artificial noise.
	The user $i\in\cM\backslash\{a\}$ produces a complex channel input of the form 
	\begin{align*}
		\bfX_i = \sum_{j=1}^{\sK}\bfPhi^{[j,i]} \bu_{j,i}^{\tau(a,i)},
	\end{align*}
	where $\tau(a,i)=a-1$ if $i<a$ and $\tau(a,i)=a$ if $i>a$, $\bu_{j,i}^{\tau(a,i)}$ is a $n^{\Gamma}\times 1$ dimension symbol vectors encoded from the confidential message $\bc_{j,i}^{\tau(a,i)}$, and $\bfPhi^{[j,i]}$ is the corresponding beamforming matrix. The transmitted signal of the user $a$ is 
	\begin{align*}
		\bfX_a = \sum_{j=1}^{\sK} \bfPhi^{[j,a]} \bm{v}_{j,a},
	\end{align*}
	where $\bm{v}_{j,a}$ is the $(n+1)^{\Gamma}\times 1$ artificial noise symbol vector chosen from the Gaussian distribution $\mathcal{CN}(0,\frac{P}{(n+1)^{\Gamma}}\mathbf{I}_{(n+1)^{\Gamma}})$, and $\bfPhi^{[j,a]}$ is the corresponding beamforming matrix.
	The channel output at server $k$ is
	\begin{align*}
		\bfY_k \!=\! \sum_{i\in\cM\backslash\{a\}} \bfH_{k,i}\!\! \sum_{j=1}^{\sK} \bfPhi^{[j,i]}\bm{u}_{j,i}^{\tau(a,i)} +\bfH_{k,a}\sum_{j=1}^{\sK}\bfPhi^{[j,a]} \bm{v}_{j,a}\! + \!\bm{Z}_k,
	\end{align*} 
	where $\bfH_{k,i}$ is a $\sTu\times\sTu$ diagonal matrix with each entry being the channel coefficient from user $i$ to server $k$ at different time slots.
	We would like all the confidential messages destined for server $j\in\cK$ to be aligned with the subspace spanned by the artificial noise $\bm{v}_{j}$ at server $k\in\cK\backslash\{j\}$, i.e., to satisfy the following $\Gamma=(\sK-1)(\sM-1)$ conditions
	\begin{align}
			\text{server } k\in\cK\backslash\{j\}:
			\begin{cases}
				\bfH_{k,1}\bfPhi^{[j,1]} \prec  \bfH_{k,a}\bfPhi^{[j,a]} \\
				\qquad \vdots \\
				\bfH_{k,a-1}\bfPhi^{[j,a-1]} \prec  \bfH_{k,a}\bfPhi^{[j,a]} \\
				\bfH_{k,a+1}\bfPhi^{[j,a+1]} \prec  \bfH_{k,a}\bfPhi^{[j,a]} \\
					\qquad \vdots \\
				\bfH_{k,\sM}\bfPhi^{[j,\sM]} \prec  \bfH_{k,a}\bfPhi^{[j,a]}  
			\end{cases}		
	\end{align}
	Let $\mbf{T}^{[k,i]} = (\mbf{H}_{k,a})^{-1}\mbf{H}_{k,i} $ and reorder all the $\mbf{T}^{[k,i]}$ by the index from $1$ to $\Gamma$. 
	For all $j\in\cK$, let $\bfPhi^{[j,1]}=\cdots=\bfPhi^{[j,a-1]}=\bfPhi^{[j,a+1]}=\cdots=\bfPhi^{[j,\sM]}$, then $\bfPhi^{[j,1]}$ and $\bfPhi^{[j,a]}$ are designed as
	\begin{align*}
		&\bfPhi^{[j,1]} = \biggl\{\Big( \prod_{i=1}^{\Gamma}\big(\mbf{T}^{[i]}\big)^{\alpha_i} \Big)\mbf{w}^{[j]}: \alpha_i\in\{1,2,...,n\} \biggr\},\\
		&\bfPhi^{[j,a]} = \biggl\{\Big( \prod_{i=1}^{\Gamma}\big(\mbf{T}^{[i]}\big)^{\alpha_i} \Big)\mbf{w}^{[j]}: \alpha_i\in\{1,2,...,n+1\} \biggr\},
	\end{align*}
	where $\mbf{w}^{[j]}$ is the $\ssT^{\prime}\times 1$ vector with each element independently chosen from a continuous distribution with bounded absolute value. 
	
	With each occupying a subspace spanned by $\bfH_{k,i}\bfPhi^{[j,a]}$, the artificial noises $\{v_{j,a}:j\in\cK\}$ occupy subspaces of a dimension of $\sK(n+1)^{\Gamma}$.
	Meanwhile, the desired messages for server $k$ form the signal space $\bfH_{k,1}\bfPhi^{[k,1]},...,\bfH_{k,a-1}\bfPhi^{[k,a-1]},\bfH_{k,a+1}\bfPhi^{[k,a+1]},\ldots,\bfH_{k,\sM}\bfPhi^{[k,\sM]}$ which are mutually independent and independent of the artificial noise \cite{wang2015secure}. 
	For $i\in\cM\backslash\{a\}$, each desired message for receiver $k\in\cK$ spans a subspace $\bfH_{k,i}\bfPhi^{[k,i]}$ of $n^{\Gamma}$ dimensions, leading to a total of $(\sM-1)n^{\Gamma}$ independent dimensions.  
	The overall $\sTu = \sK(n+1)^{\Gamma} + (\sM-1)n^{\Gamma}$ dimensions are required to successfully deliver the confidential messages. The transmission in other rounds follows the same way.
	Therefore, the sum-DoF of this channel when $\sK\geq 3$ is 
	\begin{align}\label{gen_dsum_up_a_1}
		d^{\text{up}}_{\text{sum}} = \lim_{n\to\infty}\frac{\sM (\sK (\sM-1) n^{\Gamma})}{\sM(\sK(n+1)^{\Gamma}+ (\sM-1) n^{\Gamma})}=\frac{\sK(\sM-1)}{\sK + \sM-1},
	\end{align}
	which can also be inferred from Lemma \ref{lemma_dof_1}.
	The sum-DoF for $\sK=2$ is
	\begin{align}\label{gen_dsum_up_a_2}
		d^{\text{up}}_{\text{sum}}=\frac{\sK(\sM-1)}{\sK+\sM-2},
	\end{align} 
	which can be derived using a similar but slightly different method.
	On the other hand, $\bc_{i,j}$ has a length of $\sA/r$ bits for $i\in\cM$, $j\in\cK$, according to the Definition \ref{DefDof} we also obtain 
	\begin{align}\label{gen_dsum_up_def}
		d^{\text{up}}_{\text{sum}}=\lim_{\sP\to\infty}\frac{\sK \sM \sA/r   }{ \sTu\log (\sP)}.
	\end{align} 
	Combining \eqref{gen_dsum_up_a_1}, \eqref{gen_dsum_up_a_2} and \eqref{gen_dsum_up_def}, we can get the NDT of the uplink as follows:
	\begin{align}\label{ndt_up}
		\Delta^{\text{up}} & = \lim_{\sA,\sP\to\infty} \frac{\sTu}{\sA/\log{\sP}} = \lim_{\sA,\sP\to\infty} \frac{\sK \sM}{ r \cdot d^{\text{up}}_{\text{sum}}} \nonumber\\
		& = \begin{cases}
			\frac{\sM}{ r }\frac{\sM}{\sM-1} & \sK = 2., \\
			\frac{\sK + \sM-1}{ r }\frac{\sM}{\sM-1} & \sK\geq 3.
		\end{cases}
	\end{align}
	
	After receiving the signals, each server $j\in\cK$ first recovers confidential messages $\{\bc_{j,1}, ..., \bc_{j,\sM}\}$ as in the interference alignment problem, then computes the sum of these confidential messages, i.e., 
	\begin{align}
		\mathcal{F}(\alpha_j) 
		=&\sum_{i=1}^{\sK} \bc_{j,i}= \mathcal{G}_1(\alpha_j) + ... + \mathcal{G}_\sM(\alpha_j) \nonumber\\
		=& \sum_{k=1}^{r} \bg_{D,k} \!\!\!\!  \prod_{l=1}^{[r+\sT]\backslash\{k\}} \!\! \frac{\alpha_j - \beta_{l}}{\beta_k - \beta_l} + \sum_{k=r+1}^{r+\sT} \!\! \bn_{D,k} \!\!\!\! \prod_{l=1}^{[r+\sT]\backslash\{k\}} \!\! \frac{\alpha_j - \beta_l}{\beta_k - \beta_l},\nonumber
	\end{align}
	where $\bn_{D,k} = \sum_{i=1}^{\sM}\bn_{i,k}$.
	We can see that $\mathcal{F}$ is a polynomial function of degree $r$ and $  \mathcal{F}(\beta_k)=\bg_{D,k}$ for $k\in[r]$. 
	
	For the downlink transmission, each server $j\in\cK$ generates the confidential messages $\tbc_{1,j} = \tbc_{2,j} = ... = \tbc_{\sM,j} = \mathcal{F}(\alpha_j)$, where $\tbc_{i,j}$ is intended for user $i$ for $i\in[\sM]$. To successfully send the confidential messages, the same method of noise alignment is used for the downlink. The users continue to transmit the artificial noise sequentially. The communication channel can be seen as a $\sK\times(\sM-1)$ X-network with a helper and confidential messages. 
	Similar to the uplink, $\tbc_{i,j}$ from server $j$ is split to $\sM-1$ segments each of equal size, i.e., $\tbc_{i,j}={\tbc_{i,j}^{1}, \ldots, \tbc_{i,j}^{(\sM-1)}}$. The server will send segments through $\sM-1$ rounds out of $\sM$ rounds sequentially.
	For completeness, we describe the downlink transmission at  round $a\in[\sM]$.
	The transmitted signals from the server $i\in\cK$ have the form 
	\begin{IEEEeqnarray*}{rCl}
	\tld{\mbf{X}}_i = \sum_{j\in\cM\backslash\{a\}} \tld{\bfPhi}^{[j,i]}\tld{\bm{u}}_{j,i}^{\tau(a,i)},
	\end{IEEEeqnarray*}	
	where $\tau(a,i)=a-1$ if $i<a$ and $\tau(a,i)=a$ if $i\geq a$, $\tld{\bm{u}}_{j,i}^{\tau(a,i)}$ is a $n^{\Gamma^\prime}\times 1$ symbol vectors encoded from the messages $\tld{c}_{j,i}^{\tau(a,i)}$ and $\tld{\bfPhi}^{[j,i]}$ is the corresponding beamforming matrix.
	The user $a$ is chosen to transmit artificial noise and its transmitting signals have the form 
	 \begin{align*}
	 \tld{\mbf{X}}_a = \sum_{j\in\cM\backslash\{a\}} \tld{\bfPhi}^{[j,a]}\tld{\bm{v}}_{j,a},
	 \end{align*}
	where $\tld{\bm{v}}_{j,a}$ is a $(n+1)^{\Gamma^\prime}\times 1$ artificial noise symbol vector and $\tld{\bfPhi}^{[j,a]}$ is the corresponding beamforming matrix. The received signals at a user $j\in\cM\backslash\{a\}$ are
	\begin{align*}
		\tld{\bfY}_j = \sum_{i=1}^{\sK}\tld{\bfH}_{j,i}\!\!\!\! \sum_{m\in\cM\backslash\{a\}}\!\!\!\!\tld{\bfPhi}^{[m,i]} \tld{\bu}_{m,i} + \tld{\bfH}_{j, a}\!\!\!\!\sum_{m\in\cM\backslash\{a\}}\!\!\!\!\tld{\bfPhi}^{[m,a]}\tld{\bm{v}}_m + \tld{\mbf{Z}}_j,
	\end{align*}
	where $\tld{\mbf{H}}_{j,i}$ is a $\sTd\times\sTd$ diagonal matrix with each entry being the channel coefficient from the server $i$ to a receiver $j$.
	For the confidential messages intended for user $j\in\cM\backslash\{a\}$, the following $\Gamma^{\prime} = (\sK+\sM-3)\sK$ relations are desired.
	\begin{align}
		\begin{cases}
			\text{user } m \in\cM\backslash\{j,a\}: 
			\begin{cases}
				\tld{\bfH}_{m,1}\tld{\bfPhi}^{[j,1]} \prec  \tld{\bfH}_{m,a}\tld{\bfPhi}^{[j,a]} \\
				\tld{\bfH}_{m,2}\tld{\bfPhi}^{[j,2]} \prec  \tld{\bfH}_{m,a}\tld{\bfPhi}^{[j,a]} \\
				\qquad \vdots \\
				\tld{\bfH}_{m,\sK}\tld{\bfPhi}^{[j,\sK]} \prec  \tld{\bfH}_{m,a}\tld{\bfPhi}^{[j,a]}  
			\end{cases}\\
			\text{server } s \in\cK:
			\begin{cases}
				\tld{\bfH}_{s,1}\tld{\bfPhi}^{[j,1]} \prec  \tld{\bfH}_{s,a}\tld{\bfPhi}^{[j,a]} \\
				\qquad \vdots \\
				\tld{\bfH}_{s,2}\tld{\bfPhi}^{[j,s-1]} \prec  \tld{\bfH}_{s,a}\tld{\bfPhi}^{[j,a]} \\
				\tld{\bfH}_{s,2}\tld{\bfPhi}^{[j,s+1]} \prec  \tld{\bfH}_{s,a}\tld{\bfPhi}^{[j,a]} \\
				\qquad \vdots \\
				\tld{\bfH}_{s,\sK}\tld{\bfPhi}^{[j,\sK]} \prec  \tld{\bfH}_{s,a}\tld{\bfPhi}^{[j,a]}  
			\end{cases}.
		\end{cases}
	\end{align}
	\begin{remark}
		We consider that the servers work in \emph{full-duplex} mode, i.e., the server can transmit and receive the signals simultaneously. Therefore, to prevent a server from attaining additional information from the signals transmitted by the other servers in the downlink, the alignment conditions at the servers are needed. If the servers work in the \emph{half-duplex} mode, i.e., the servers cannot transmit and receive the signals simultaneously, the alignment conditions on the servers can be removed, then $\Gamma^{\prime} = (\sM-2)\sK$.
	\end{remark}

	Similar to the uplink, let $\tld{\bfPhi}^{[j,1]}=\cdots={\bfPhi}^{[j,a-1]}={\bfPhi}^{[j,a+1]}=\ldots=\tld{\bfPhi}^{[j,\sM]}$ for $j\in\cM\backslash\{a\}$, $\tld{\bfPhi}^{[j,1]}$ and $\tld{\bfPhi}^{[j,a]}$ are designed as
	\begin{align*}
		&\tld{\bfPhi}^{[j,1]} = \biggl\{\Big( \prod_{i=1}^{\Gamma}\big(\tld{\mbf{T}}^{[i]}\big)^{\alpha_i} \Big)\tld{\mbf{w}}^{[j]}: \alpha_i\in\{1,2,...,n\} \biggr\},\\
		&\tld{\bfPhi}^{[j,a]} = \biggl\{\Big( \prod_{i=1}^{\Gamma}\big(\tld{\mbf{T}}^{[i]}\big)^{\alpha_i} \Big)\tld{\mbf{w}}^{[j]}: \alpha_i\in\{1,2,...,n+1\} \biggr\},
	\end{align*}
	where $\tld{\mbf{T}}^{[k,i]} = (\tld{\mbf{H}}_{k,a})^{-1}\tld{\mbf{H}}_{k,i} $ are reordered by the index from $1$ to $\Gamma^{\prime}$, and $\tld{\mbf{w}}^{[j]}$ is the $\sTd\times 1$ vector with each element independently chosen from a continuous distribution with bounded absolute value. 	
	For $\sM\geq 3$, the sum-DoF of the downlink channel is 
	\begin{IEEEeqnarray*}{rCl}
	d^{\text{down}}_{\text{sum}} = \frac{ \sK}{\sM + \sK-1}.
	\end{IEEEeqnarray*}
	Also, $\tbc_{i,j}$ has a length of $\sA/r$ bits for $i\in\cM$, $j\in\cK$, according to the Definition \ref{DefDof} we have $d^{\text{down}}_{\text{sum}} = \lim_{P\to\infty}\frac{\sK\sA/r}{\sTd\log(\sP)}$.
	Therefore, we can get the NDT of the downlink is
	\begin{align}\label{ndt_down}
		\Delta^{\text{down}} & = \lim_{\sA,\sP\to\infty} \frac{\sTd}{\sA/\log{\sP}} = \lim_{\sA,\sP\to\infty} \frac{\sK }{ r \cdot d^{\text{down}}_{\text{sum}}} \nonumber\\
		& = \frac{\sM + \sK-1}{ r }.
	\end{align} 
	After the downlink transmission, each user attains $\sK$ distinct points $(\alpha_i, \mathcal{F}(\alpha_i))_{i\in\cK}$ of the polynomial $\mathcal{F}(x)$ of degree $r$, and thus can construct the function $\mathcal{F}(x)$ by interpolation if   $\sK\geq r+\sT $.  Finally, each user recovers the global gradient $\bg_D=(\bg_{D,1},\ldots,\bg_{D,r})$ from evaluations $\{\bg_{D,k} = \mathcal{F}(\beta_k):k\in[r]\}$.
	
	From \eqref{ndt_up} and\eqref{ndt_down}, we know that the achievable NDT of the proposed scheme is as the following theorem.
	\begin{theorem}\label{theorem:achNDT}
		For the wireless private gradient aggregation system with $\sM\geq 3$ users and $\sK\geq 2$ servers, the following uplink and downlink NDTs are achievable:
		\begin{IEEEeqnarray}{rCl}\label{eqAchNDT}
			\Delta^{\tn{up}} \!=\! \begin{cases}
				\frac{\sM}{r} \frac{\sM}{\sM-1}, & \sK\!=\!2 \\
				\frac{\sK+\sM-1}{r} \frac{\sM}{\sM-1}, & \sK\!\geq\! 3  
			\end{cases},\ \ 	
			\Delta^{\tn{down}} \!=\!\frac{\sK+\sM-1}{r} ,
		\end{IEEEeqnarray}
		where $r+1\leq \sK$.
	\end{theorem}
	\begin{proof}
		See the scheme present above.
	\end{proof} 
	
	Recall that $r$ is the number of partitions of gradients in Lagrange coding.  From \eqref{eqAchNDT}, we can see that the uplink and downlink NDTs both decrease with $r$.
	When $\sK\geq 3$ and $r=\sK-1$, the uplink and downlink NDTs can be written as $(\frac{\sM}{\sK-1}+1)\frac{\sM}{\sM-1}$ and $\frac{\sM}{\sK-1}+1$, respectively. It is observed that the uplink and downlink NDTs decrease monotonically with the number of servers $\sK$. The uplink NDTs increase monotonically with the number of users when $\sM>\sqrt{\sK}+1$, while the downlink NDT always increases monotonically with the number of users.
	Note that the uplink and downlink NDTs can be written as $\frac{\sM\sK}{r\cdot d^{\text{up}}_{\text{sum}}}$ and $\frac{\sK}{r\cdot d^{\text{down}}_{\text{sum}}}$ as shown in \eqref{ndt_up} and \eqref{ndt_down}. This indicates that interference alignment technology, which allows multiple users to send messages to multiple servers simultaneously, can reduce uplink communication latency.
	
	\begin{remark}
		Recall that the degree of \eqref{gen_encode_Fi} is $r+1$, it is sufficient to recover the aggregation value by only $r+1$ points $(\alpha_i, \mathcal{F}(\alpha_i))_{i\in\mathcal{R^{\prime}}}$, where $\mathcal{R^{\prime}}\subset\cK$ and $|\mathcal{R^{\prime}}|=r+1$. If $r< \sK-1$, it is not necessary to let all $\sK$ servers participate in the downlink transmission. If $s$ servers do not participate in the downlink transmission where $s<\sK-r-1$, the downlink NDT is $\frac{\sM+\sK-s-1}{r}\geq \frac{\sM+\sK-s-1}{\sK-s-1}=\frac{\sM}{\sK-s-1}+1$. Although the NDT becomes larger, the property demonstrates the resistance of the secret-sharing method to the straggler problem, which has been shown in other scenarios.
	\end{remark}
	

	\subsubsection{Comparison with single server systems}\label{rem: single}
	By comparison, consider a single-server system with $\sM$ users. In the uplink, the time-division multiple access (TDMA) is applied for transmitting local updates from the users to the server. For the downlink, the server directly broadcasts the aggregated results to all users.
	The corresponding sum-DoF of the uplink channel and downlink channel are $d^{\tn{up}}_{s} = 1 = \frac{\sM\sA}{\ssT^{\tn{u}} \log(\sP)}$ and $d^{\tn{down}}_{s} = 1 = \frac{\sA}{\ssT^{\tn{d}} \log(\sP)}$, respectively. We can attain the NDTs in the single server system are
	\begin{align}\label{eq: single}
		\Delta^{\tn{up}}_{\tn{s}} = \frac{\ssT^{\tn{u}}}{\sA/\log(P)} = \sM \text{\ and\ } \Delta^{\tn{down}}_{\tn{s}} = \frac{\ssT^{\tn{d}}}{\sA/\log(P)} = 1.
	\end{align}
	
	Note that the benefits of introducing multiple servers in our scheme are two-fold: First, the multiple-server system can ensure the aggregation results to be accessible to the users while remaining private to the server, while the single-server system cannot achieve this goal without incorporating additional encryption techniques. Second, our multi-server scheme improves communication efficiency, i.e., achieving a lower uplink NDT except the case $\sK=2$ (see Fig.~\ref{res01}), which is expected as we additionally account for the privacy constraint. 
	{Fig.~\ref{res01} compares the uplink and downlink NDTs versus $\sM$, including the NDTs of the proposed scheme in Theorem \ref{theorem:achNDT} for $\sK = 2,4$ and $8$, as well as the NDT of single server system described as \eqref{eq: single}.
	We observe that both the uplink and downlink NDTs increase monotonically over most of the range of $\sM$, while decreasing with $\sK$. Additionally, as $\sK$ increases, the growth rate of the uplink and downlink NDTs with respect to $\sM$ slows down. This indicates that increasing the number of servers is beneficial for the scalability of the aggregation systems.}
	\begin{figure}[tb]
		\centering
		\includegraphics[width=\linewidth]{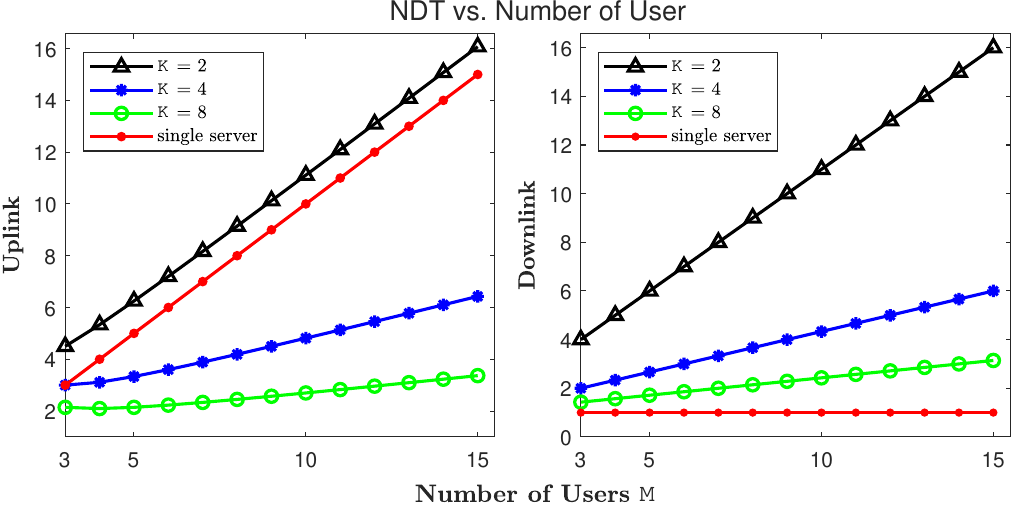}
		\caption{{Uplink and Downlink NDTs versus $\sM$ when $r=\sK-1$.}  }
		\label{res01}
	\end{figure}

	\subsubsection{Complexity of the proposed scheme}
		{The complexity of the proposed scheme mainly comes from the coding of secret sharing method and the design of beamforming and decoding matrices for noise alignment.}
 		Complexity of the secret sharing: For each user, generating the secret shares using \eqref{gen_encode_Gi} can be viewed as interpolating polynomials of degree $r$, and evaluating them at $\sK$ points. There are efficient algorithms with almost linear computational complexities. It is shown that interpolating a polynomial of degree $r$ and evaluating it at arbitrary $\sK$ points have complexities of $O(r\log^2 r\log\log r )$ and $O(\sK\log^2 \sK\log\log \sK )$, respectively \cite{kedlaya2011fast}. Thus, the complexity for each user to generate the secret shares is $O(\sK\log^2 \sK\log\log\sK\cdot\frac{p}{r})$. Similarly, the user can decode the gradient segments from the aggregated secret shares using polynomial interpolation and evaluation, which has a complexity of $O(r\log^2 r\log\log r \cdot\frac{p}{r})$. Note that the users can generate the random noise in \eqref{gen_encode_Gi} separately, and only need to agree on the elements $\{\beta_{l},\cdots,\beta_{r+1},\alpha_1,\cdots,\alpha_\sK\}$ where the resulting communication overhead is negligible.
		 Complexity of noise alignment: In the delivery phase, each transmitter needs beamforming matrices of size $\ssT\times n^{\Gamma}$ for each of its transmitted messages in $\ssT$ transmission timeslots with $n\in\mathbb{N}^{+}$, and the beamforming matrices design are given in Section \ref{general description}. The complexity for each transmitter is given by $O(\sK^{2}(\sM+\sK)\ssT n^{(\sM+\sK)\sK})$. After receiving all signals in $\ssT$ transmission timeslots, each receiver needs to compute an inversion of a matrix of size $\ssT\times \ssT$ and its complexity is given by $O(\ssT^{3})$. 
		\begin{remark}
		To reduce the complexity of interference alignment and remove the need for CSI at transmitters (CSIT), blind artificial noise alignment that combines blind interference alignment with artificial noise transmission can be applied \cite{gou2011aiming,wang2015secure}. The transmission scheme incorporating blind interference alignment follows a similar principle to our scheme presented above but assumes that each receiver is equipped with a single antenna that can switch between multiple predefined modes. This can significantly reduce computational complexity and mitigate the challenges associated with the acquisition of CSIT. The application of this approach to distributed computation systems without privacy consideration, such as wireless MapReduce, has been explored in \cite{huang2024coded,lu2024blind}. 
		\end{remark}

	\section{Optimality and Privacy Analysis}
	\label{sec:opt_priv analysis}
	In this section, we analyze the optimality and the privacy guarantee of the proposed scheme. 
	For ease of description, we use the following notations to denote the sets of messages, local randomness, and signals:
	\begin{IEEEeqnarray*}{rCl}
	 &&\bc_{\cK,\cM}=\{\bc_{i,j}:i\in\cK,j\in\cM\},  \\
	 &&\tld{\bc}_{\cM,\cK} = \{\tld{\bc}_{j,i}: j\in\cM, i\in \cK\}, \\
	&&\bm{u}_{\cK,\cM} = \{\bm{u}_{j,i}: j\in\cK, i\in \cM\},  \\
	&&\bm{v}_{\cK}=\{\bm{v}_{j,a}: j\in\cK\}, \quad \tld{\bm{v}}_{\cM}=\{\tld{\bm{v}}_{j,a}: j\in\cM\backslash\{a\}\},\\
	&&\bfY_{\cK}=\{ \bfY_j: j\in\cK \}, \quad \tilde{\bfY}_{\cM}=\{ \tilde{\bfY}_j: j\in\cM \}, \\
	&&\bg_{\cM}=\{\bg_i: i\in\cM\}, \quad \bn_{\cM} = \{\bn_{i}:i\in\cM\}
\end{IEEEeqnarray*}
	
	\subsection{Converse bound}
	For the considered system, we have an information-theoretic lower bound as the following theorem.
	\begin{theorem}\label{theorem:conv}
		For the wireless private gradient aggregation system with $\sM$ users and $\sK$ servers, the uplink and downlink NDTs are lower-bounded as:
		\begin{IEEEeqnarray}{rCl}\label{eqConvNDT}
			\Delta^{\tn{up}}_{\tn{lb}} = \frac{\max\{\sM,\sK\} }{\sK-1},\ \ 	
			\Delta^{\tn{down}}_{\tn{lb}} = \frac{\sK}{\sK-1},
		\end{IEEEeqnarray}
	\end{theorem}
	\begin{proof}
		See Appendix \ref{sec:convproof}.
	\end{proof}

	\begin{theorem}\label{corol:gap}
		For the considered aggregation system, let $r=\sK-1$. The uplink NDT are asymptotically optimal if $\sK\gg\sM$ and $\sM\gg 0$, and the downlink NDT are asymptotically optimal if $\sK\gg\sM$. Moreover, the multiplicative gap between the achievable uplink NDT and the optimal uplink NDT is no more than $4$, regardless of  $\sK$ and $\sM$.
	\end{theorem}
		\begin{proof}
			Let $r=\sK-1$, the attain uplink NDT of the proposed scheme is $\Delta^{\tn{up}} = \begin{cases}
				\frac{\sM}{\sK-1}\frac{\sM}{\sM-1}, & \sK=2 \\
				(\frac{\sM}{\sK-1}+ 1)\frac{\sM}{\sM-1}, & \sK\geq 3  
			\end{cases}$, 
			and $\Delta^{\text{down}} = \frac{\sM +\sK-1}{\sK-1}$ for $\sM\geq 3$.
			We can observe that 
			\begin{align*}
				\frac{\Delta^{\text{up}} }{\Delta_\text{lb}^{\text{up}}} &\leq \frac{\sK+\sM}{\max\{\sM,\sK\}}\frac{\sM}{\sM-1} \\
				&\leq \frac{2\cdot\max\{\sM,\sK\}}{\max\{\sM,\sK\}}\frac{\sM}{\sM-1} \leq \frac{2\sM}{\sM-1}\leq 4.
			\end{align*}
			And $\frac{\Delta^{\text{up}} }{\Delta_\text{lb}^{\text{up}}}\approx 1$ if $\sK\gg\sM\gg 0$.
			Meanwhile, if $\sK\gg\sM$, $\frac{\Delta^{\text{down}} }{\Delta_\text{lb}^{\text{down}}} \leq \frac{\sK+\sM}{\sK} \approx 1 $. This completes the proof.
		\end{proof}}
	\begin{remark}(Communication cost)
		Our scheme has the minimal communication costs, which are measured by the bit length of the transmitted messages. More specifically, when $r=\sK-1$, the uplink and downlink communication costs are $H(\bc_{\cK,\cM})=\frac{\sK}{\sK-1}\sM\sA$ bits and $H(\tld{\bc}_{\cM,\cK})=\frac{\sK}{\sK-1}\sA$ bits, respectively. 
		When $\sK$ is relatively large, $H(\bc_{\cK,\cM})\approx\sM\sA$ bits and $H(\tld{\bc}_{\cM,\cK})\approx\sA$ bits which are $\sM$ times and equal to the bit length of the aggregation $\gD$, respectively, showing that the communication costs are minimal.
		The reason is that to achieve a privacy-guaranteed aggregation of all local gradients, each local value must be uploaded at least once when there is no side information in the servers, and the downlink transmissions should contain at least the same amount of information as the aggregation value. Although the communication costs are asymptotically minimal, the privacy constraints and the channel interference reduce the DoF, resulting in an increment in the NDTs compared to the lower bound.
	\end{remark}
	\begin{figure}[tb]
		\centering
		\subfloat[Uplink]{
			\includegraphics[width=0.9\columnwidth]{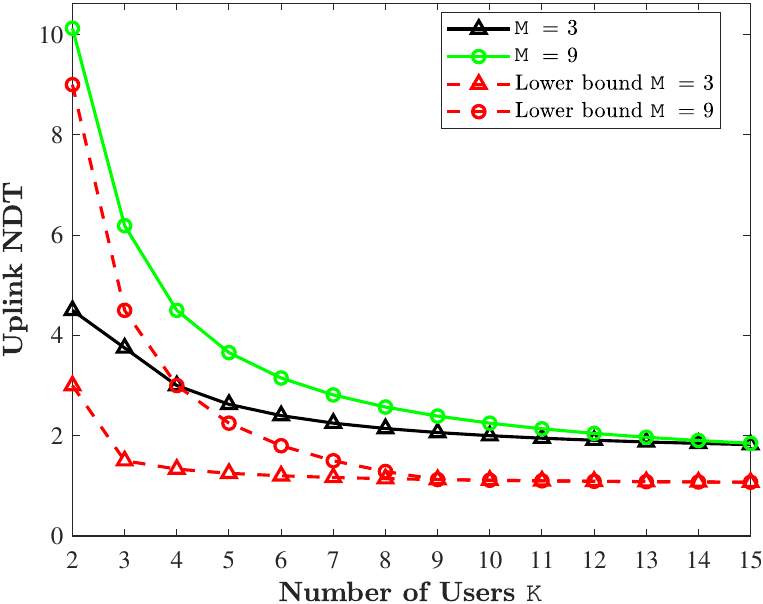}
		}\\
		\subfloat[Downlink]{
			\includegraphics[width=0.9\columnwidth]{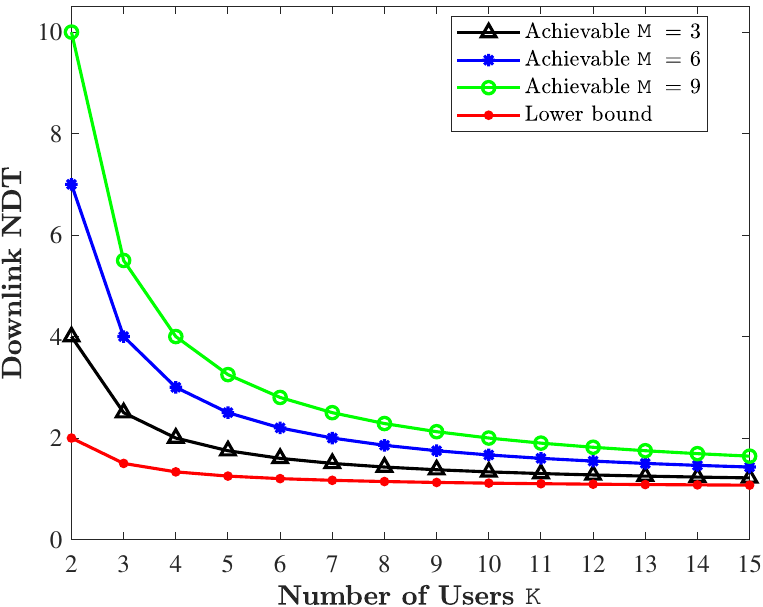}
		}
		\caption{The comparison of achievable NDT and the converse bound when $r=\sK-1$ and the number of users $\sM$ are $3$, $6$ and $9$.  }
		\label{res02}
	\end{figure}
	Fig.~\ref{res02} compares the achievable NDT of the proposed scheme and the lower bound when setting $r=\sK-1$. As can be seen from the Corollary \ref{corol:gap}, the uplink NDT of our scheme is order-optimal. In particular, the gap between the achievable uplink NDT and the optimum becomes smaller than $4$ when $\sK$ and $\sM$ have a relatively large difference. The achievable downlink NDT is also order-optimal when the number of servers is larger than the number of users, but it is not satisfactory when $\sK$ is small. {It can be seen that both the achievable NDT and the lower bound decrease with $\sK$.} Moreover, when $\sK$ is relatively large, both the uplink NDT and the downlink NDT are close to their respective lower bounds, indicating that increasing the number of servers is beneficial for reducing communication latency.

	\subsection{Privacy Analysis}
	\label{subsec:priv}
	In the following,  we would like to show that the information leakage at the server can be bounded.
	
	Firstly, the equivocation for $\bg_{\cM}$ at a server $j$ is
	\begin{align*}
		\frac{H(\bg_{\cM}|\bfY_{j}, \tilde{\bfY}_j )}{H(\bg_{\cM})} = 1 - \frac{ I(\bg_{\cM};\bfY_{j},\tilde{\bfY}_j ) }{ H(\bg_{\cM}) }.
	\end{align*}
	Since $H(\bg_\cM)$ is constantly equal to $\sM\cdot\sA$ bits, it is needed to bound the mutual information term. 
	$I(\bg_{\cM};\bfY_{j},\tilde{\bfY}_j )$ can be bounded as 
	\begin{align} \label{general_equivocation_1}
		\lefteqn{I(\bg_{\cM};\bfY_{j},\tilde{\bfY}_j )} \nonumber\\
		&= I(\bg_{\cM};\bfY_{j} ) + I(\bg_{\cM};\tilde{\bfY}_j| \bfY_{j} ) \nonumber\\
		& \leq I(\bg_{\cM};\bfY_{j},\bm{u}_{j,\cM}) + I(\bg_{\cM};\tilde{\bfY}_j| \bfY_{j} )\nonumber \\
		& = I(\bg_{\cM};\bm{u}_{j,\cM}) + I(\bg_{\cM};\bfY_{j}|\bm{u}_{j,\cM}) + I(\bg_{\cM};\tilde{\bfY}_j| \bfY_{j} )
	\end{align}
	
	The first mutual information term in \eqref{general_equivocation_1} is bounded as  
	\begin{align}
		I(\bg_{\cM};\bm{u}_{j,\cM}) &\leq I(\bg_{\cM};\bc_{j,\cM}) \label{generalI1_process}\\
		&= H(\bc_{j,\cM} ) - H( \bc_{j,\cM}| \bg_{\cM} ) \nonumber\\
		&= H(\bc_{j,\cM} ) - H( U_{j}^L\cdot \tilde{\bg} + U_{j}^R\cdot\tilde{\bn}| \bg_{\cM}) 
		\label{generalI1_n}\\
		&= H(\bc_{j,\cM} ) - H( U_{j}^R\cdot\tilde{\bn} ) = 0, \label{general_pr_I1eq0}
	\end{align}
	where $U_{j}^L$, $U_{j}^R$ denote the submatrices consisting of the rows of $U^L$, $U^R$ indexed by $j$.
	Here, \eqref{generalI1_process} comes from data processing inequality and definition of mutual information; \eqref{generalI1_n} follows from \eqref{largrang_matrix}; \eqref{general_pr_I1eq0} holds because $\tilde{\bn}$ is independent of $\bg_{\cM}$, and $\bc_{j,\cM}$ is in $\mathbb{F}^{1\times \sM\times\frac{p}{r}}$, $U_{j}^R\cdot\tilde{\bn}$ has a random and uniform distribution in $\mathbb{F}^{1\times \sM\times\frac{p}{r}}$. 
	
	{In the following, the analysis is tailored to the transmission at round $a$ but omits the index of transmission round. The results hold throughout the whole transmission because of the independence among different transmission rounds.} 
	The second mutual information term in \eqref{general_equivocation_1} is bounded as
	\begin{align}
		\lefteqn{I( \bg_{\cM};\bfY_{j} | \bm{u}_{j,\cM} )} \nonumber\\
		&\leq I( \bg_{\cM}, \bm{u}_{\cK\backslash \{j\},\cM};\bfY_{j} | \bm{u}_{j,\cM} ) \nonumber\\
		&= I( \bm{u}_{\cK\backslash \{j\}, \cM};\bfY_{j} | \bm{u}_{j,\cM} ) + I( \bg_{\cM};\bfY_{j} | \bm{u}_{\cK,\cM} ) \label{general_I2_chainrule_1} \\
		&= I( \bm{u}_{\cK\backslash\{j\},\cM}, \bm{v}_{\cK};\bfY_{j} | \bm{u}_{j,\cM} ) - I( \bm{v}_{\cK};\bfY_{j}|\bm{u}_{\cK,\cM} ) \label{general_I2_chainrule_2}
	\end{align} 
	where \eqref{general_I2_chainrule_1} comes from the chain rule and \eqref{general_I2_chainrule_2} comes from the chain rule and the fact that $\bg_{\cM}$ can be recovered from $\bm{u}_{\cK,\cM}$.
	
	{Consider the received signals at server $j$, we have \begin{align}
		\mbf{Y}_j = \mbf{D}_j \begin{bmatrix}
			\bm{u}_{j,\cM\backslash\{a\}} \\\bm{v}_{j,a}
		\end{bmatrix} + \sum_{k\in\cK\backslash\{j\}}\mbf{Q}^j_k \begin{bmatrix}
			\bm{u}_{k,\cM\backslash\{a\}} \\\bm{v}_{k,a}
		\end{bmatrix},
	\end{align}
	where 	$\mbf{D}_j = \big[ \bfH_{j,1}\bfPhi^{[j,1]}  \ \dots\   \bfH_{j,\sM}\bfPhi^{[j,\sM]}\ \bfH_{j,a}\bfPhi^{[j,a]} \big]$, $\mbf{Q}^j_{k} = \big[ \bfH_{j,1}\bfPhi^{[k,1]}  \ \dots\  \bfH_{j,\sM}\bfPhi^{[k,\sM]} \  \bfH_{j,a}\bfPhi^{[k,a]} \big]$	and $\begin{bmatrix} \bm{u}_{k,\cM\backslash\{a\}} \\\bm{v}_{k,a} \end{bmatrix} = \big[ \bm{u}_{k,1}^{T}\  \cdots\  \bm{u}_{k,\sM}^T\ \bm{v}_{k,a}^T \big]^{T}$. }
	According to the design of the beamforming matrices $\bfPhi^{[j,i]}$s and $\bfPhi^{[j,a]}$, the transform matrix $\mbf{Q}^j_{k}$ has the following property
	\begin{align*}
		\text{span}(\mbf{Q}^j_{k}) \prec \text{span}( \bfH_{j,a}\bfPhi^{[k,a]}).
	\end{align*} 
	This means that the artificial noise dominates every dimension of the subspace spanned by the private messages \cite{wang2015secure}. Let $\bar{\mbf{Q}}^j_k = \bfH_{j,a}\bfPhi^{[k,a]}$, continuing the proof in \eqref{general_I2_chainrule_2}, we get
	\begin{align}
		\lefteqn{I( \bg_{\cM};\bfY_{j}|\bm{u}_{j,\cM} )} \nonumber\\
		&\leq I( \bm{u}_{\cK\backslash\{j\},\cM}, \bm{v}_{\cK};\bfY_{j} | \bm{u}_{j,\cM} ) - I( \bm{v}_{\cK};\bfY_{j}|\bm{u}_{\cK,\cM} ) \nonumber\\
		&\leq \log \det\Big(\sP\cdot \!\!\!\! \sum_{k\in\cK}\bar{\mbf{Q}}^j_k(\bar{\mbf{Q}}^j_k)^H + \mbf{I} \Big) \nonumber\\
		&\hspace{1cm} - \log\det \Big(\sP\cdot \!\!\!\! \sum_{k\in\cK}\bar{\mbf{Q}}^j_k(\bar{\mbf{Q}}^j_k)^H + \mbf{I} \Big)\nonumber + o(\log(\sP))  \\
		& = \sK(n+1)^{\Gamma}\log(\sP) - \sK(n+1)^{\Gamma}\log(\sP) + o(\log(\sP)) \nonumber\\
		& = o(\log(\sP)). \label{general_pr_I2eq0}
	\end{align}
	
	Similarly, the third mutual information term in \eqref{general_equivocation_1} is bounded as  
	\begin{align}
		\lefteqn{I(\bg_{\cM};\tilde{\bfY}_j| \bfY_{j} )} \nonumber\\
		&= I(\bg_{\cM}, \tilde{\bm{v}}_\cM;\tilde{\bfY}_j| \bfY_{j} ) - I(\tilde{\bm{v}}_\cM;\tilde{\bfY}_j| \bfY_{j}, \bg_{\cM} ) \nonumber\\
		&\leq H(\tilde{\bfY}_j| \bfY_{j}) - H(\tilde{\bfY}_j| \bfY_{j}, \bg_{\cM}, \tilde{\bm{v}}_\cM) \nonumber\\
		&\hspace{1.5cm}- H(\tilde{\bm{v}}_\cM| \bfY_{j}, \bg_{\cM}) + H(\tilde{\bm{v}}_\cM| \bfY_{j}, \bg_{\cM}, \tilde{\bfY}_j) \nonumber\\
		&\leq H(\tilde{\bfY}_j) - H(\tilde{\bfY}_j| \bfY_{j}, \bg_{\cM}, \tilde{\bm{u}}_{\cM,\cK}, \tilde{\bm{v}}_\cM) \nonumber\\
		&\hspace{1.5cm}- H(\tilde{\bm{v}}_\cM) + H(\tilde{\bm{v}}_\cM | \tilde{\bfY}_j ) \label{general_3_tldv}\\
		&\leq H(\tilde{\bfY}_j) - H(\tilde{\bfY}_j| \tilde{\bm{u}}_{\cM,\cK}, \tilde{\bm{v}}_\cM) - H(\tilde{\bm{v}}_\cM) + H(\tilde{\bm{v}}_\cM | \tilde{\bfY}_j ) \label{general_3_mar}\\
		&\leq I(\tilde{\bm{u}}_{\cM,\cK}, \tilde{\bm{v}}_\cM;\tilde{\bfY}_j ) - I(\tilde{\bm{v}}_\cM; \tilde{\bfY}_j ) \nonumber\\
		&\leq \log\det\Big(\sP\cdot\!\!\!\! \sum_{m\in\cM\backslash\{a\}} \tilde{\mbf{Q}}_m^j (\tilde{\mbf{Q}}_m^j)^H + \mbf{I} \Big) \nonumber\\
		&\hspace{1cm} - \log\det\Big(\sP\cdot\!\!\!\! \sum_{m\in\cM\backslash\{a\}} \bar{\tilde{\mbf{Q}}}_m^j (\bar{\tilde{\mbf{Q}}}_m^j)^H + \mbf{I} \Big) + o(\log(\sP))\nonumber\\
		&= o(\log(\sP)) \label{general_pr_I3eq0},
	\end{align}
	where \eqref{general_3_tldv} holds since conditioning reduces entropy and $\tld{\bm{v}}_\cM$ is the random artificial noise that is independent with $\bfY_j$ and $\bg_{\cM}$, \eqref{general_3_mar} holds since $\bfY_j$ and $\bg_{\cM}$ are independent with $\tilde{\bfY}_j$ given $\tilde{\bm{u}}_{\cM,\cK}$ and $\tilde{\bm{v}}_\cM$. $\tilde{\mbf{Q}}_m^j, \bar{\tilde{\mbf{Q}}}_m^j$ have the similar form to $\mbf{Q}_m^{j},\bar{\mbf{Q}}_m^j$ but act in the downlink.
	
	Substituting \eqref{general_pr_I1eq0}, \eqref{general_pr_I2eq0}, and \eqref{general_pr_I3eq0} into \eqref{general_equivocation_1}, it is readily shown that, $\forall j\in\cK$, 
	\begin{IEEEeqnarray}{rCl}
		 \lim_{A,P\to\infty}  \Delta_{\bg_{\cM}}^{j} &=&   \lim_{A,P\to\infty}  \frac{H(\bg_{\cM}|\bfY_{j},\tilde{\bfY}_j)}{H(\bg_\mathcal{M})} \nonumber\\
		 &=& \lim_{A,P\to\infty}   \big(1 -\frac{o(\log{P})}{\sM\sA} \big)\nonumber\\&=&  1 \nonumber.
	\end{IEEEeqnarray}

	\section{Conclusion}
	\label{sec:conclusion}
	{
		In this work, we considered the secure gradient aggregation problem for multi-server wireless federated learning systems.
		We proposed a privacy-preserving coded aggregation scheme and characterized the uplink and downlink communication latency measured by the NDT. Our scheme guarantees that each server cannot infer either the local gradients or the aggregation value in an information-theoretic sense. We also established lower bounds on the optimal uplink and downlink NDT and theoretically proved that the proposed scheme is asymptotically optimal when the number of servers and users is sufficiently large. 
		}

	\begin{appendices}

	\section{Proof of the Theorem \ref{theorem:conv}}
	\label{sec:convproof}
	
	We first present the converse proof for the uplink NDT. 
	According to the chain rule of mutual information, we have the following relation
	\begin{align}
		&H(\bg_D| \gmt,\nmt ) = H(\bg_D | \bc_{\cK,\cM}, \gmt,\nmt) \nonumber\\
		& \hspace{2.5cm} + I(\bg_D; \bc_{\cK,\cM}|\gmt,\nmt) \label{cp:up-chain-1}.
	\end{align} 
	From the correctness constraint, the first term of the right-hand side of \eqref{cp:up-chain-1} can be bounded as
	\begin{align}
		\lefteqn{H(\gD | \bc_{\cK,\cM}, \gmt,\nmt)}\nonumber\\
		&= H(\gD)- I(\gD;\bc_{\cK,\cM},\gmt,\nmt) \nonumber\\
		&= H(\gD)- I(\gD;\bc_{\cK,\cM},\gD, \gmt,\nmt) \nonumber\\
		&\leq H(\gD)- I(\gD;\bc_{\cK,\cM},\bg_i) \nonumber\\
		&\leq H(\gD)- I(\gD;(\tld{c}_{i,1},\cdots,\tld{c}_{i,\sK}),\bg_{{i}}) \label{cp:up-tcikfunofckm}\\
		&= 0 \label{cp:up-correct_cons},
	\end{align}
	where \eqref{cp:up-tcikfunofckm} holds because $(\tld{c}_{i,1},\cdots,\tld{c}_{i,\sK})$ is determined by $\bc_{\cK,\cM}$, and \eqref{cp:up-correct_cons} comes from the correctness constraint \eqref{cons:correct}.
	Besides, from the secrecy constraints we get
	\begin{IEEEeqnarray}{rCl}
		\lefteqn{I(\gD; \bc_{j,\cM}|\gmt,\nmt)} \nonumber\\
		&=& H(\gD| \gmt,\nmt) - H(\gD|\bc_{j,\cM}, \gmt,\nmt) \nonumber\\
		&=& H(\bg_i) - H(\bg_i|\bc_{j,i}) \label{cp:up-priv_gn-mutual_indep}\\
		&=& I(\bg_i; \bc_{j,i}) \leq I(\bg_i; \bfY_{j}) \label{cp:up-priv_cjidecbyyj} \\
		&=& \epsilon_1 \label{cp:up-priv_giyj},
	\end{IEEEeqnarray}
	where \eqref{cp:up-priv_gn-mutual_indep} follows from the mutual independence of the gradients and the local randomness, \eqref{cp:up-priv_cjidecbyyj} holds since $\bc_{j,i}$ can be decoded from $\bfY_{j}$, and \eqref{cp:up-priv_giyj} comes from privacy constraint, i.e., Definition \ref{DefPrivacy}, and $\epsilon_1$ vanishes as $\sA,\sP$ tends to infinity. 
	
	With \eqref{cp:up-correct_cons} and \eqref{cp:up-priv_giyj}, we can continue the \eqref{cp:up-chain-1} as follows
	\begin{IEEEeqnarray}{rCl}
		\lefteqn{H(\bg_D| \gmt,\nmt )} \nonumber\\
		& =& I(\gD;\bc_{\cK,\cM} | \gmt,\nmt) \nonumber\\
		&& \hspace{1cm}- I(\gD; \bc_{j,\cM}| \gmt, \nmt) + \epsilon_1 \nonumber\\
		& =& I(\gD; \bc_{\cK\backslash\{j\},\cM} | \bc_{j,\cM}, \gmt, \nmt) + \epsilon_1 \nonumber \\
		& =& H(\bc_{\cK\backslash\{j\},\cM}| \bc_{j,\cM}, \gmt,\nmt) \nonumber\\
		&& \qquad - H(\bc_{\cK\backslash\{j\},\cM} | \bc_{j,\cM}, \gD,\gmt,\nmt) + \epsilon_1 \nonumber\\
		& \leq& H(\bc_{\cK\backslash\{j\},\cM}| \bc_{j,\cM}, \gmt,\nmt) + \epsilon_1 \nonumber\\
		& =& H(\bc_{\cK\backslash\{j\}, i}| \bc_{j,i}) + \epsilon_1 \nonumber\\
		&{\leq} & \frac{\sK-1}{\sK}\sum_{k=1}^{\sK}H(\bc_{k,i}) +\epsilon_1, \label{cp:up-ine_cki}
	\end{IEEEeqnarray}
	where \eqref{cp:up-ine_cki} comes from the \emph{Han's inequality}.
	
	From \eqref{cp:up-ine_cki} we can get
	\begin{IEEEeqnarray}{cCl}
		\frac{\sM\cdot\sK}{\sK-1}\cdot \sA -\epsilon_2 &\leq& \sum_{i=1}^{\sM}\sum_{k=1}^{\sK}H(\bc_{k,i}) \leq I(\bfX_{\cM}; \bfY_{\cK})\nonumber \\
		&\leq& \min\{\sM,\sK\}\cdot \sTu\cdot(\log(\sP)+o(1))\label{cp:up-mimocap},
	\end{IEEEeqnarray}
	where $\epsilon_2$  vanishes as $\sA,\sP$ go to infinity, and \eqref{cp:up-mimocap} comes from {the results on the capacity of the multiple-input-multiple-output channels}.
	Then it is readily shown that 
	\begin{align}
		\lim_{A,P\to\infty}\frac{\sTu}{\sA/\log(\sP)} \geq \frac{\max\{\sM,\sK\}}{\sK-1}.
	\end{align}
	
	Next, we present the converse proof for the downlink NDT. 
	Begin with the chain rule of the mutual information, for $\forall t\neq i \neq j$, the following relations hold.
	\begin{IEEEeqnarray}{rCl}
		\lefteqn{H(\gD | \gmt,\nmt)} \nonumber\\
		&=& H(\gD| \tilde{\bc}_{t, \cK}, \gmt,\nmt) \nonumber\\
		&&\hspace{1cm}+ I(\gD; \tilde{\bc}_{t,\cK}| \gmt,\nmt) \nonumber\\
		&\leq& H(\gD| \tbc_{t, \cK}, \bg_{\cM\backslash\{i\}}) + I(\gD; \tilde{\bc}_{t,\cK}| \gmt,\nmt) \nonumber\\
		& = & 0 + I(\gD; \tilde{\bc}_{t,\cK\backslash\{j\}}| \tilde{\bc}_{t,j}, \gmt,\nmt) \nonumber\\
		&& \hspace{1cm}+ I(\gD; \tilde{\bc}_{t,j}| \gmt,\nmt) \label{cp:down-correct_cons}
	\end{IEEEeqnarray}
	where \eqref{cp:down-correct_cons} comes from the correctness constraint \eqref{cons:correct} and the chain rule.
	On the other hand, the following relationship holds
	\begin{align}
		\lefteqn{I(\gD; \tbc_{t,j}| \gmt,\nmt)} \nonumber\\
		&= I(\gD;\tbc_{t,j}, \bc_{j,\cM}| \gmt,\nmt) \nonumber\\
		& \qquad -I(\gD; \bc_{j,\cM}| \tbc_{t,j},\gmt,\nmt) \nonumber\\
		&= I(\gD;\tbc_{t,j} | \bc_{j,\cM}, \gmt,\nmt) \nonumber\\
		& \qquad + I(\gD; \bc_{j,\cM}| \gmt,\nmt) \nonumber\\
		& \qquad  -I(\gD; \bc_{j,\cM}| \tbc_{t,j},\gmt,\nmt) \nonumber\\
		&= I(\gD; \bc_{j,\cM}| \gmt,\nmt) \nonumber\\
		& \qquad -I(\gD; \bc_{j,\cM}| \tbc_{t,j},\gmt,\nmt) \label{cp:tcc} \\
		&\leq  I(\gD; \bc_{j,\cM}| \gmt,\nmt) \overset{\eqref{cp:up-priv_giyj}}{\leq} \epsilon_1 \label{cp:down-priv},
	\end{align}
	where $\eqref{cp:tcc}$ holds because $\tbc_{t,j}$ is a function of $\bc_{j,\cM}$. 
	Substituting the \eqref{cp:down-priv} into the \eqref{cp:down-correct_cons}, we can get
	\begin{IEEEeqnarray}{rCl}
		\lefteqn{H(\gD | \gmt,\nmt)} \nonumber\\
		&=& I(\gD; \tilde{\bc}_{t,\cK\backslash\{j\}}| \tilde{\bc}_{t,j}, \gmt,\nmt) + \epsilon_1 \nonumber\\
		&=& H(\tilde{\bc}_{t,\cK\backslash\{j\}}| \tilde{\bc}_{t,j},\gmt,\nmt) \nonumber\\
		&& \qquad - H(\tilde{\bc}_{t,\cK\backslash\{j\}}| \tilde{\bc}_{t,j},\gD, \gmt,\nmt) +\epsilon_1 \nonumber\\
		&\leq&  H(\tilde{\bc}_{t,\cK\backslash\{j\}}| \tilde{\bc}_{t,j}) + \epsilon_1 \nonumber\\
		&\leq & \frac{\sK-1}{\sK} \sum_{k=1}^{\sK} H(\tilde{\bc}_{t,k}) + \epsilon_1 \label{cp:down-ine_tcki}
	\end{IEEEeqnarray}
	\eqref{cp:down-ine_tcki} comes from the \emph{Han's inequality}.

	From \eqref{cp:down-ine_tcki} we can attain 
	\begin{IEEEeqnarray}{cCl}
		\frac{\sK}{\sK-1}\cdot \sA -\epsilon_3 &\leq& \sum_{k=1}^{\sK}H(\tbc_{t,k}) \nonumber \\
		&\leq& I(\tilde{\bfX}_{\cK}; \tilde{\bfY}_{t} | \bg_{t} ) \leq H(\tilde{\bfY}_{t})\nonumber\\
		&\leq& \sTd\cdot(\log(\sP)+o(1))\label{cp:maccap},
	\end{IEEEeqnarray}
	where $\epsilon_3$ vanishes as $\sA,\sP$ go to infinity.
	Then it is readily shown that
	\begin{align}
		\lim_{A,P\to\infty}\frac{\sTd}{\sA/\log(\sP)} \geq \frac{\sK}{\sK-1}.
	\end{align}
	This completes the proof.
	
	\end{appendices}
	\ifCLASSOPTIONcaptionsoff
	\newpage
	\fi


	
	%
	\bibliographystyle{IEEEtran}
	\bibliography{ref.bib}
	
	%

	
	

\end{document}